\pdfoutput=1

\RequirePackage{fix-cm,amsmath}
\documentclass[%
  a4paper,
  natbib,
  smallextended,
]{svjour3}

\smartqed  % flush right qed marks, e.g. at end of proof

\usepackage[utf8]{inputenc}
\usepackage[T1]{fontenc}
\usepackage[english]{babel}
\usepackage{microtype}
\usepackage{capt-of}
\usepackage{algorithm}  \floatname{algorithm}{\footnotesize Algorithm}
\usepackage{algpseudocode}
\usepackage{graphicx}
\usepackage{bm}
\usepackage{booktabs}
\usepackage{nicefrac}
\usepackage{float}
\usepackage{nicefrac}
\usepackage{multirow}
\usepackage{mathptmx}% use Times fonts if available on your TeX system
\usepackage{pifont}
\usepackage{geometry} 
\usepackage{hyperref}

\algnewcommand\algorithmicforeach{\textbf{foreach}}
\algdef{S}[FOR]{ForEach}[1]{\algorithmicforeach\ #1\ \algorithmicdo}
\newcommand\abbr[1]{\textsc{\MakeLowercase{#1}}}
\newcommand\kNN{\ensuremath{k}\abbr{NN}}
\let\dataset\abbr
\newcommand*{\tran}{^{\raisebox{-0.1ex}{$\scriptscriptstyle\mkern-1.5mu\mathsf{T}$}}}
\renewcommand{\S}{\mathbf{S}}
\newcommand{\E}{\mathbf{E}}
\newcommand{\e}{\mathbf{e}}
\renewcommand{\O}{\mathcal{O}}
\newcommand*{\BETA}{_{\raisebox{0ex}{$\scriptscriptstyle\mkern-1.5mu\bm\beta$}}}

\newcommand*{\true}{\ensuremath{\text{\ding{51}}}}
\newcommand*{\bftrue}{\ensuremath{\text{\ding{52}}}}
\newcommand*{\false}{\ensuremath{\text{\ding{55}}}}
\newcommand{\orcid}[1]{\href{https://orcid.org/#1}{\raisebox{-0.3ex}{\protect\includegraphics[height=1em]{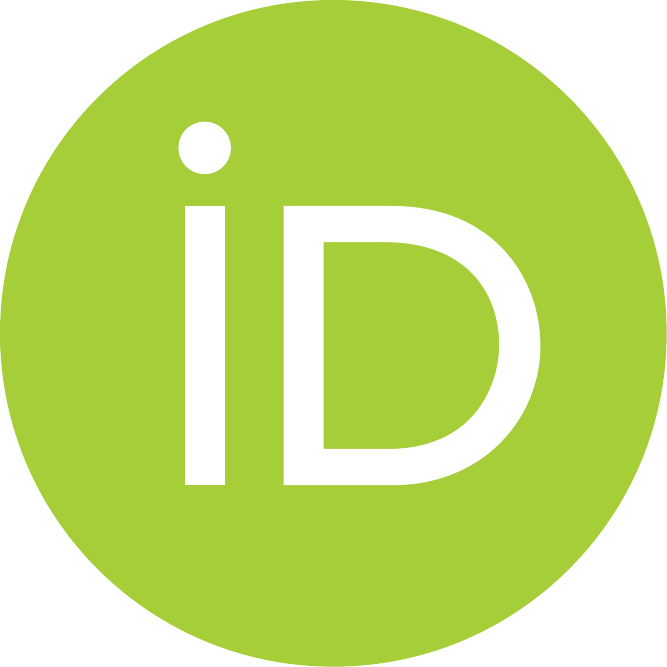}}}}
\newenvironment{funding}{\begingroup\def\ackname{Funding}\acknowledgements}{\endacknowledgements\endgroup}
\newenvironment{conflict}{\begingroup\def\ackname{Conflict of Interest}\acknowledgements}{\endacknowledgements\endgroup}

\begin{document}

\title{Text classification with word embedding regularization and soft similarity measure}
\subtitle{}


\author{
  Vít Novotný~\orcid{0000-0002-3303-4130} \and
  Eniafe Festus Ayetiran~\orcid{0000-0002-6816-2781} \and
  Michal Štefánik~\orcid{0000-0003-1766-5538} \and
  Petr Sojka~\orcid{0000-0002-5768-4007}%
}

\authorrunning{Vít Novotný et al.} % if too long for running head

\institute{
  Vít Novotný \and
  Eniafe F. Ayetiran \and
  Michal Štefánik \and
  Petr Sojka \at
  Faculty of Informatics
  Masaryk University,
  Botanická 68a,
  Brno 602\,00,
  Czechia\\
  Tel.: +420-774\,488\,765\\
  \email{witiko@mail.muni.cz}%
}


\maketitle

\begin{abstract}
% Situation
Since the seminal work of \citeauthor{mikolov2013efficient}, word embeddings
have become the preferred word representations for many natural language
processing tasks.  Document similarity measures extracted from
word embeddings, such as the soft cosine measure (\abbr{SCM}) and the Word
Mover's Distance (\abbr{WMD}), were reported to achieve state-of-the-art
performance on the semantic text similarity and text classification.

% Problem
Despite the strong performance of the \abbr{WMD} on text classification
and semantic text similarity, its
super-cubic average time complexity is impractical.
The \abbr{SCM} has quadratic worst-case time complexity, but its
performance on text classification has never been compared with the \abbr{WMD}.
Recently, two word embedding regularization techniques were shown to reduce
storage and memory costs, and to improve training speed, document processing speed, and task performance on
word analogy, word similarity, and
semantic text similarity.  However, the effect of these techniques on text
classification has not yet been studied.

% Solution
In our work, we investigate the individual and joint effect of the two word
embedding regularization techniques on the document processing speed and the task performance of the \abbr{SCM}
and the \abbr{WMD} on text classification. For evaluation, we use the
$k$\abbr{NN} classifier and six standard datasets:
\dataset{BBCSPORT}, \dataset{TWITTER}, \dataset{OHSUMED},
\dataset{REUTERS-21578}, \dataset{AMAZON}, and \dataset{20NEWS}.
\looseness=-1

% Evaluation
We show 39\% average $k$\abbr{NN} test error reduction with regularized word 
embeddings compared to non-regularized word embeddings.
We describe a practical procedure for deriving such regularized
embeddings through Cholesky factorization. 
We also show that the \abbr{SCM} with regularized word embeddings significantly
outperforms the \abbr{WMD} on text classification and is over $10{,}000\times$ faster.

% Metadata
\keywords{%
  Text classification \and
  Soft Cosine Measure \and
  Word Mover's Distance \and
  Word embedding regularization%
}
%\subclass{
%  15A63 \and
%  11Y05 \and
%  68T50%
%}
\end{abstract}

\section{Introduction}
\label{intro}
Word embeddings are the state-of-the-art words representation for many
natural language processing (\abbr{NLP}) tasks. Most of these
tasks are at the word level, such as the word
analogy~\citep{garten2015combining} and word similarity,
and at the sentence level, such as question answering, natural language
inference, semantic role labeling, co-reference resolution, named-entity
recognition, and sentiment analysis~\citep{DBLP:conf/naacl/PetersNIGCLZ18}, semantic text
similarity~\citep{charlet2017simbow}. On document-level tasks, such as
machine translation~\citep{qi2018and}, text
classification~\citep{kusner2015word, Wuetal2018} and ad-hoc information
retrieval~\citep{Zucconetal, kuzi2016query}, word embeddings provide simple
and strong baselines.

Document similarity measures, such as the Soft Cosine Measure
(\abbr{SCM})~\citep{sidorov2014soft, charlet2017simbow, vit} and the Word
Mover's Distance (\abbr{WMD})~\citep{kusner2015word} can be extracted from word
embeddings. The \abbr{SCM} achieves state-of-the-art performance on
the semantic text similarity task~\citep{charlet2017simbow}. The \abbr{WMD}
outperforms standard methods, such as the \abbr{VSM}, \abbr{BM25}, \abbr{LSI},
\abbr{LDA}, m\abbr{SDA}, and \abbr{CCG} on the text classification
task~\citep{kusner2015word}, and achieves
state-of-the-art performance on nine text classification datasets and 22
semantic text similarity datasets with better performance on datasets with shorter
documents. The \abbr{SCM} is asymptotically faster than the \abbr{WMD}, but their
task performance has never been compared.

Regularization techniques were reported to improve the task performance of word
embeddings. We use
the quantization technique of \citet{lam}, which reduces storage and memory cost,
and improves training speed~\citep{courbariaux2016binarized} and task performance on
word analogy and word similarity. We also use the orthogonalization technique
of \citet{vit}, which improves the document processing speed and the task performance of the \abbr{SCM} on semantic text
similarity~\citep{charlet2017simbow}. However, the effect of these techniques
at the document-level (e.g.\ text classification) has not been studied.

In our work, we investigate the effect of word embedding regularization on text
classification. The contributions of our work are as follows: (1)~We show that
word embedding regularization techniques that reduce storage and memory costs
and improve speed can also significantly improve performance on the text
classification task. (2)~We show that the \abbr{SCM} with regularized word
embeddings significantly outperforms the slower \abbr{WMD} on the text
classification task. (3)~We define orthogonalized word embeddings and we prove
a relationship between orthogonalized word embeddings, Cholesky factorization,
and the word embedding regularization technique of \citet{vit}.

The rest of the paper is organized as follows: We present related work in
Section~\ref{sec:related-work}. In Section~\ref{sec:vector-space-models}, we
discuss the document distance and similarity measures. In
Section~\ref{sec:regularization}, we discuss word embedding regularization
techniques and we prove their properties.
Section~\ref{sec:experiment} presents our experiment and
Section~\ref{sec:results} discusses our results.
Section~\ref{sec:conclusion} concludes the paper.

\section{Related Work}
\label{sec:related-work}
Word embeddings represent words in a vector space, where syntactically and
semantically similar words are close to each other. Word embeddings can be
extracted from word co-occurrence matrices~\citep{deerwester1990indexing,
pennington2014glove} and from neural network language models~\citep{bengio2003neural,
mikolov2013efficient, DBLP:conf/naacl/PetersNIGCLZ18}.
Word embeddings extracted from neural network language models have been shown to be effective on
several (\abbr{NLP}) tasks, but they tend to suffer from overfitting due to high
feature dimensionality. There have been several works that use word embedding
regularization to reduce overfitting.

\citet{IgorandHod} introduce a technique which re-embeds an existing embedding
with the end product being a target embedding. In their technique, they perform
optimization of a convex objective. Their objective is a linear combination of
the log-likelihood of the dataset under a designated target embedding and
the Frobenius norm of a distortion matrix. The Frobenius norm serves as a
regularizer that penalizes the Euclidean distance between the initial and the
designated target embeddings. To enrich the embedding, they further use
external source embedding, which they incorporated into the regularizer, on the
supervised objective. Their approach was reported to improve performance on the
sentiment classification task.

The Dropout technique was introduced by \citet{Nitishetal} to
mitigate the problem of overfitting in neural networks by dropping their
neurons and links between neurons during
training. During the training, Dropout samples an exponential number
of the reduced networks and at test time, it approximates the effect of averaging
the previsions of these thinned networks by using a single unthinned
network with smaller weights. Learning the Dropout networks involves two major
steps: back propagation and unsupervised pre-training. Dropout was successfully applied to
a number of tasks in the areas of vision, speech recognition, document
classification, and computational biology.

\citet{sunetal} introduce a sparse
constraint into Word2Vec~\citep{mikolov2013efficient} in order to increase its
interpretability and performance. They added the $\ell_1$ regularizer into the
loss function of the Continuous Bag-of-Words (\abbr{CBOW}).
They applied the technique to online learning and to solve the
problem of stochastic gradient descent, they employ an online optimization
algorithm for regularized stochastic learning -- the Regularized Dual Averaging
(\abbr{RDA}).
\looseness=-1

In their own work, \citet{pengetal} experimented with four
regularization techniques: penalizing weights (embeddings
excluded), penalizing embeddings, word re-embedding and Dropout. At the end of
their experiments, they concluded the following: (1)~regularization techniques do
help generalization, but their effect depends largely on the dataset size.
(2)~penalizing the $\ell_2$-norm of embeddings also improves task performance
(3)~incremental hyperparameter tuning achieves
similar performance, which shows that regularization has mostly a local
effect (4)~Dropout performs slightly worse than the $\ell_2$-norm penalization
(5)~word re-embedding does not improve task performance.
\looseness=-1

Another approach by
\citet{songetal} uses pre-learned or external priors as a regularizer for the
enhancement of word embeddings extracted from neural network language models.
They considered two types of
embeddings. The first one was extracted from topic distributions generated from
unannotated data using the Latent Dirichlet Allocation (\abbr{LDA}). The second was
based on dictionaries that were created from human annotations. A novel data
structure called the trajectory softmax was introduced for effective learning
with the regularizers. \citeauthor{songetal} reported improved embedding quality through
learning from prior knowledge with the regularizer.

A different algorithm was
presented by \citet{yangetal}. They applied their own regularization to cross-domain
embeddings. In contrast to \citet{sunetal}, who applied their regularization to
the \abbr{CBOW}, the technique of \citeauthor{yangetal} is a regularized skip-gram
model, which allows word embeddings to be learned from different domains. They
reported the effectiveness of their approach with experiments on entity
recognition, sentiment classification and targeted sentiment analysis.

\citet{berend} in his own approach investigates the effect of $\ell_1$-regularized
sparse word embeddings for identification of multi-word expressions.
\citeauthor{berend} uses dictionary learning, which decomposes the original
embedding matrix by solving an optimization problem.

Other works that focus solely on reducing storage and memory costs of word
embeddings include the following:
\citet{hinton2014distilling} use a distillation
compression technique to compress the knowledge in an ensemble into a single
model, which is much easier to deploy.
\citet{chen2015compressing} present
HashNets, a novel framework to reduce redundancy in neural networks.
Their neural architecture
uses a low-cost hash function to arbitrarily group link weights into hash
buckets, and all links within the same hash bucket share a single parameter
value.
\citet{see2016compression} employ a network weight pruning strategy and apply it to
Neural Machine Translation (\abbr{NMT}). They experimented with three
\abbr{NMT} models, namely the class-blind, the class-uniform and the class-distribution model.
The result of their experiments shows that even strong weight pruning does not
reduce task performance.
FastTest.zip is a compression technique by
\citet{joulin2016fasttext1} who use product quantization to mitigate the loss
in task performance reported with earlier quantization techniques.
In an attempt to reduce the space and memory costs of word embeddings,
\citet{shu2017compressing}
experimented with construction of embeddings with a few basis vectors, so
that the composition of the basis vectors is determined by a hash code.

Our technique is based partly on the recent regularization technique by
\citet{lam}, in which a quantization function was introduced into the loss function
of the \abbr{CBOW} with negative sampling to show
that high-quality word embeddings using 1--2 bits per parameter can be learned.
\citeauthor{lam}'s technique is based on the work of \citet{courbariaux2016binarized}, who employ
neural networks with binary weights and activations to reduce
space and memory costs. A major component of their work is the use of bit-wise
arithmetic operations during computation.

\section{Document Distance and Similarity Measures}
\label{sec:vector-space-models}
The Vector Space Model (\abbr{VSM})~\citep{ml:SaltonBuckley1988} is a
distributional semantics model that is fundamental to a number of text
similarity applications including text classification. The \abbr{VSM}
represents documents as coordinate vectors relative to a real
inner-product-space orthonormal basis $\bm\beta$, where coordinates correspond
to weighted and normalized word frequencies. In the \abbr{VSM}, a commonly used
measure of similarity for document vectors $\mathbf x$ and $\mathbf y$ is the
cosine similarity:
\begin{equation}
  \text{cosine similarity of }\mathbf x\text{ and }\mathbf y =
  \langle\mathbf x/\Vert\mathbf x\Vert_2, \mathbf y/\Vert\mathbf y\Vert_2\rangle
  \text{, where }\langle\mathbf x, \mathbf y\rangle = \big((\mathbf
  x)\BETA\big)\tran(\mathbf y)\BETA\text{ and }\Vert\mathbf z\Vert_2\text{ is
  the }\ell_2\text{-norm of }\mathbf z.
\end{equation}

The cosine similarity is highly susceptible to polysemy, since distinct words
correspond to mutually orthogonal basis vectors. Therefore, documents that use
different terminology will always be regarded as dissimilar. To borrow an
example from \citet{kusner2015word}, the cosine similarity of the documents
``Obama speaks to the media in Illinois'' and ``the President greets the press
in Chicago'' is zero if we disregard stop words.

The Word Mover's Distance (\abbr{WMD}) and the Soft Cosine Measure (\abbr{SCM})
are document distance and similarity measures that address polysemy.  Because
of the scope of this work, we discuss briefly the \abbr{WMD} and the \abbr{SCM}
in the following subsections.

\subsection{Word Mover's Distance}
The Word Mover's Distance (\abbr{WMD})~\citep{kusner2015word} uses network
flows to find the optimal transport between \abbr{VSM} document vectors. The
distance of two document vectors $\mathbf x$ and $\mathbf y$ is the following:
\begin{equation}
  \text{\abbr{WMD}}(\mathbf x, \mathbf y) = \text{minimum cumulative cost }
  \sum_{i,j}f_{ij}c_{ij}\text{ of a flow }\mathbf F=(f_{ij})\text{ subject to }
  \mathbf F\geq 0, \sum_j f_{ij}=(x_i)\BETA,
\end{equation}
where the cost $c_{ij}$ is the Euclidean distance of embeddings for words $i$
and $j$. We use the implementation in PyEMD~\citep{pele2008, pele2009} with
the best known average time complexity $\O(p_{\mathbf{xy}}^3\log
p_{\mathbf{xy}})$, where $p_{\mathbf{xy}}$ is the number of unique words in
$\mathbf x$ and~$\mathbf y$.

\subsection{Soft Cosine Measure}
\label{sec:scm}
The soft \abbr{VSM}~\citep{sidorov2014soft, vit} assumes that document vectors
are represented in a non-orthogonal normalized basis $\bm\beta$. In the soft
\abbr{VSM}, basis vectors of similar words are close and the cosine similarity
of two document vectors $\mathbf x$ and $\mathbf y$ is the Soft Cosine
Measure~(\abbr{SCM}):
\begin{equation}
  \label{eq:scm}
  \text{\abbr{SCM}}(\mathbf x, \mathbf y) = \langle\mathbf x/\Vert\mathbf
  x\Vert_2, \mathbf y/\Vert\mathbf y\Vert_2\rangle
  \text{, where }\langle\mathbf x, \mathbf y\rangle = \big((\mathbf
  x)\BETA\big)\tran\S(\mathbf y)\BETA\text{, and }\S\text{ is a word similarity matrix.}
\end{equation}
We define the word similarity matrix $\S$ like
\citet{charlet2017simbow}: $s_{ij} = \max(t, \langle\e_i/\Vert\e_i\Vert_2,
\e_j/\Vert\e_j\Vert_2\rangle)^o\!\!$, where $\e_i$ and $\e_j$ are the
embeddings for words $i$ and $j$, and $o$ and $t$ are free parameters.
We use the implementation in the \texttt{similarities.termsim} module of
Gensim~\citep{rehurek10framework}.
The worst-case time complexity of the \abbr{SCM} is $\O(p_{\mathbf x}p_{\mathbf
y})$, where $p_{\mathbf x}$ is the number of unique words in $\mathbf x$ and
$p_{\mathbf y}$ is the number of unique words in $\mathbf y$.

\section{Word Embedding Regularization}
\label{sec:regularization}
The Continuous Bag-of-Words Model (\abbr{CBOW})~\citep{mikolov2013efficient}
is a neural network language model that predicts the center word from context
words. The \abbr{CBOW} with negative sampling minimizes the following loss
function:
\looseness=-1
\begin{equation}
  J(\mathbf u_o, \hat{\mathbf v}_c) = -\log\big(\sigma(\langle\mathbf u_o, \hat{\mathbf v}_c\rangle)\big)
  - \sum_{i = 1}^k \log\big(\sigma(-\langle\mathbf u_i, \hat{\mathbf v}_c\rangle)\big),
  \text{where }\hat{\mathbf v}_c = \frac{1}{2w}\sum_{-w+i\leq i\leq w+o, i\neq o} \mathbf v_i,
\end{equation}
$\mathbf u_o$ is the vector of a center word with corpus position $o$,
$\mathbf v_i$ is the vector of a context word with corpus position $i$,
and the window size $w$ and the number of negative samples $k$ are free parameters.

Word embeddings are the sum of center word vectors and context word vectors.
To improve the properties of word embeddings and the task performance of the
\abbr{WMD} and the \abbr{SCM}, we apply two regularization techniques to
\abbr{CBOW}.
\looseness=-1

\subsection{Quantization}
Following the approach of \citet{lam}, we quantize the center word vector $\mathbf u_o$ and the context word
vector $\mathbf v_i$ during the forward and backward propagation stages of the training:
\begin{equation}
  \text{a quantized center word vector }\mathbf u_o = \nicefrac{1}{3}\cdot\text{sign}(\mathbf u_o)
  \text{ and a quantized context word vector }\mathbf v_i = \nicefrac{1}{3}\cdot\text{sign}(\mathbf v_i).
\end{equation}
Since the quantization function is non-differentiable at certain points, we
use \citeauthor{hinton2012neural}'s straight-through
estimator~\citep[Lecture~15b]{hinton2012neural} as the gradient:
\begin{equation}
  \nabla(\nicefrac{1}{3}\cdot\text{sign}) = \nabla I,
  \text{where }\nabla\text{ is the gradient operator and }I\text{ is the identity function.}
\end{equation}

Lam shows that quantization reduces the storage and memory cost and improves
the performance of word embeddings on the word analogy and word similarity
tasks.

\begin{figure*}
\hspace{-0.4cm}%
\includegraphics{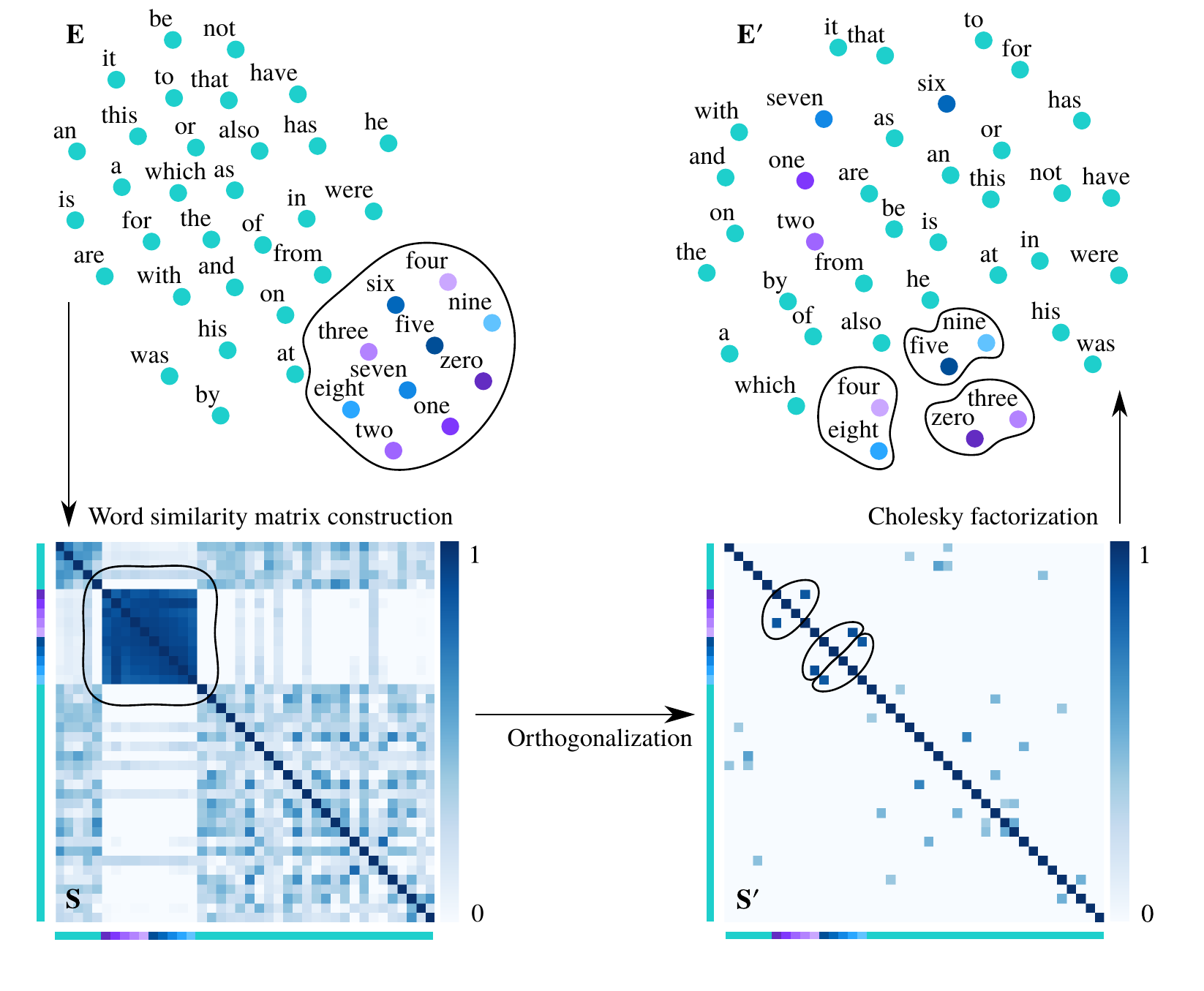}%
\caption{%
  T-\abbr{SNE} visualizations of non-regularized word embeddings $\E$ (top left)
  and orthogonalized word embeddings $\E'$
  (top right) for the 40~most common words on the first 100~MiB of
  the English Wikipedia, and images of the word similarity matrices $\S$ (bottom left) and
  $\S'$ (bottom right).
  Word similarity matrix construction uses the parameters $o = 1$ and $t = -1$, and orthogonalization
  uses the parameter values $\text{Sym} = \text{\true}, \text{Dom} = \text{\true},$ and
  $\text{Idf} = \text{\true}$.
  Notice how the cluster of numerals from $\E$ is separated in $\E'$ due to the parameter
  value $\text{Idf} = \text{\true}$, which makes common words more likely to be mutually orthogonal.
}
\label{fig:orthogonalization}
\end{figure*}

\subsection{Orthogonalization}
\label{sec:orthogonalization}
\citet{vit} shows that producing a sparse word similarity matrix $\S'$ that
stores at most $C$ largest values from every column of $\S$ reduces the
worst-case time complexity of the \abbr{SCM} to $\O(p_{\mathbf x})$, where
$p_{\mathbf x}$ is the number of unique words in a document vector $\mathbf x$.

\citeauthor{vit} also claims that $\S'$ improves the performance of the soft
\abbr{VSM} on the question answering task and describes a greedy algorithm for
producing $\S'$, which we will refer to as the orthogonalization algorithm.
The orthogonalization algorithm has three boolean parameters: Sym, Dom, and
Idf. Sym and Dom make $\S'$ symmetric and strictly diagonally dominant. Idf
processes columns of $\S$ in descending order of inverse document
frequency~\citep{robertson2004understanding}:
\begin{equation}
  \text{inverse document frequency of word }w = -\log_2\text P(w\mid D)
  = \log_2\frac{|D|}{|\{d\in D\mid w\in d\}|}
  \text{, where }D\text{ are documents.}
\end{equation}

In our experiment, we compute the \abbr{SCM} directly from the word similarity
matrix $\S'$, see Equation~\eqref{eq:scm}.  However, actual word embeddings
must be extracted for many \abbr{NLP} tasks. \citeauthor{vit} shows that the
word similarity matrix $\S'$ can be decomposed using Cholesky factorization. We
will now define orthogonalized word embeddings and we will show that the
Cholesky factors of $\S'$ are in fact orthogonalized word embeddings.

\begin{definition}[Orthogonalized word embeddings]
Let $\E, \E'$ be real matrices with $|V|$ rows, where $V$ is a vocabulary of
words. Then $\E'$ are orthogonalized word embeddings from $\E$, which we denote
$\E'\leq_\bot\E$, iff for all $i,j=1,2,\ldots,|V|$ it holds that $\langle\e'_i, \e'_j\rangle
\neq0\implies\langle\e'_i, \e'_j\rangle=\langle\e_i, \e_j\rangle$, where $\e_k$
and $\e'_k$ denote the $k$-th rows of $\E$ and~$\E'$.
\end{definition}

\begin{theorem}
\label{thm:orthogonalization}
Let $\E$ be a real matrix with $|V|$ rows, where $V$ is a vocabulary of words,
and for all $k=1,2,\ldots,|V|$ it holds that $\Vert\e_k\Vert_2=1$.
Let $\S$ be a word similarity matrix constructed from $\E$ with the parameter
values $t=-1$ and $o=1$ as described in Section~\ref{sec:scm}. Let $\S'$ be a
word similarity matrix produced from $\S$ using the orthogonalization algorithm
with the parameter values $\text{Sym}=\text{\true}$ and
$\text{Dom}=\text{\true}$.  Let $\E'$ be the Cholesky factor of $\S'$. Then
$\E'\leq_\bot\E$.
\end{theorem}

\begin{proof}
With the parameter values $\text{Sym}=\text{\true}, \text{Dom}=\text{\true}$,
$\S'$ is symmetric and strictly diagonally dominant, and therefore also
positive definite. The symmetric positive definite matrix~$\S'$ has a unique
Cholesky factorization of the form $\S' = \E'(\E')\tran$. Therefore, the
Cholesky factor $\E'$ exists and is uniquely determined.

From $\S' = \E'(\E')\tran$, we have that for all $i,j=1,2,\ldots,|V|$ such that
the sparse matrix $\S'$ does not contain the value $s'_{ij}$ it holds that
$s'_{ij} = \langle\e'_i, \e'_j\rangle = 0$. Since the implication in the
theorem only applies when $\langle\e'_i, \e'_j\rangle\neq0$, we do not need to
consider this case.

From $\S' = \E'(\E')\tran, o=1, t=-1,$ and $\Vert\e_k\Vert_2=1$, we have that
for all $i,j=1,2,\ldots,|V|$ such that the sparse matrix $\S'$ contains the value
$s'_{ij}$, it holds that $\langle\e'_i, \e'_j\rangle = s'_{ij} = s_{ij} =
\max(t, \langle\e_i/\Vert\e_i\Vert_2,
\e_j/\Vert\e_j\Vert_2\rangle)^o=\langle\e_i,\e_j\rangle.$\qed
\end{proof}

Figure~\ref{fig:orthogonalization} shows the extraction of orthogonalized word
embeddings $\E'$ from $\E$: From $\E$, we construct the dense word similarity
matrix $\S$ and from $\S$, we produce the sparse word similarity matrix $\S'$
through orthogonalization. From $\S'$, we produce the orthogonalized
embeddings $\E'$ through Cholesky factorization.

With a large vocabulary~$V$, a $|V|\times|V|$ dense matrix~$\S$ may not fit
in the main memory, and we produce the sparse matrix~$\S'$ directly
from $\E$. Similarly, a $|V|\times|V|$ dense matrix~$\E'$ may also not fit
in the main memory, and we use sparse Cholesky factorization to produce a
sparse matrix~$\E'$ instead. If dense word embeddings are required, we
use dimensionality reduction on $\E'$ to produce a $|V|\times D$ dense matrix,
where $D$ is the number of dimensions.

With the parameter value $\text{Idf}=\text{\true},$ words with small inverse
document frequency, i.e.\ common words such as numerals, prepositions, and
articles, are more likely to be mutually orthogonal (i.e. $\langle\e'_i,
\e'_j\rangle=0$) than rare words. This is why in
Figure~\ref{fig:orthogonalization}, the numerals form a cluster in the
non-regularized word embeddings $\E$, but they are separated in orthogonalized word
embeddings $\E'$.

\section{Experiment}
\label{sec:experiment}
The experiment was conducted using six~standard text classification datasets by employing both the
\abbr{WMD} and the \abbr{SCM} with a $k$ Nearest Neighbor (\kNN{}) classifier using both regularized
and non-regularized word embeddings. First, we describe briefly our datasets, then
our experimental steps. Our experimental code is available online.%
\footnote{See \url{https://github.com/MIR-MU/regularized-embeddings}.}

\subsection{Datasets}
In our experiment, we used the following six~standard text classification datasets:
\paragraph{BBCSPORT} 
The \dataset{BBCSPORT} dataset~\citep{greene06icml}
consists of 737~sport news articles from the \abbr{BBC} sport website
in five topical areas: athletics, cricket, football, rugby, and
tennis. The period of coverage was during 2004--2005.

\paragraph{TWITTER}
The \dataset{TWITTER} dataset~\citep{sanders2011sanders} consists of 5,513
tweets hand-classified into one of four topics: Apple, Google,
Twitter, and Microsoft. The sentiment of every tweet was also hand-classified as
either Positive, Neutral, Negative, or Irrelevant.

\paragraph{OHSUMED}
The \dataset{OHSUMED} dataset~\citep{Hersh} is a set of 348,566 references spanning 1987--1991
from the \abbr{MEDLINE} bibliographic database of important, peer-reviewed medical
literature maintained by the National Library of Medicine (\abbr{NLM}). While the
majority of references are to journal articles, there are also a small number
of references to letters of the editor, conference proceedings, and other
reports. Each reference contains human-assigned subject headings from the
17,000-term Medical Subject Headings (MeSH) vocabulary.
\looseness=-1

\paragraph{REUTERS}
The documents in the \dataset{REUTERS}-21578 collection~\citep{lewis1997reuters}
appeared on the Reuters newswire in 1987.  The documents were assembled and
indexed with categories by the personnel of Reuters Ltd. The collection is
contained in 22~\abbr{SGML} files. Each of the first 21~files (\texttt{reut2-000.sgm}
through \texttt{reut2-020.sgm}) contains 1,000 documents, while the last one (\texttt{reut2-021.sgm})
contains only 578 documents.

\paragraph{AMAZON}
The \dataset{AMAZON} dataset~\citep{mcauley2015image} contains 142.8
million product reviews from Amazon spanning 1996--2014. Each product belongs to one of 24~categories,
which include Books, Cell Phones \& Accessories, Clothing, Shoes \& Jewelry,
Digital Music, and Electronics, among others. The $5$-core subset of the
\abbr{AMAZON} dataset contains only those products and users with at least five
reviews.

\paragraph{20NEWS}
The \abbr{20NEWS} dataset~\citep{Lang95} is a collection of 18,828 Usenet newsgroup
messages partitioned across 20 newsgroups with different topics. The collection
has become popular for experiments in text classification and text clustering.

\subsection{Preprocessing}
For \dataset{TWITTER}, we use 5,116 out of 5,513 tweets due to unavailability,
we restrict our experiment to the Positive, Neutral, and Negative classes, and
we subsample the dataset to 3,108 tweets like \citet{kusner2015word}.
For \dataset{OHSUMED}, we use the 34,389 abstracts related to cardiovascular
diseases~\citep{joachims1998text} out of 50,216 and we restrict our experiment
to abstracts with a single class label from the first 10~classes out of~23. In
the case of \dataset{REUTERS}, we use the R8 subset~\citep{cachopo2007improving}.
For \dataset{AMAZON}, we use 5-core reviews from the Books, CDs and Vinyl,
Electronics, and Home and Kitchen classes.

We preprocess the datasets by lower-casing the text and by tokenizing to
longest non-empty sequences of alphanumeric characters that contain at least
one alphabetical character. We do not remove stop words or rare words, only
words without embeddings. We split each dataset into train and test subsets
using either a standard split (for \dataset{REUTERS} and \dataset{20NEWS}) or
following the split size of \citet{kusner2015word}. See
Table~\ref{table:dataset-statistic} for statistics of the preprocessed datasets.

\begin{figure}
\vspace{-0.8cm}
\begin{minipage}[b]{0.48\textwidth}
  \begin{algorithm}[H]
    \caption{%
      \footnotesize\enskip
      The \kNN{} classifier for the \abbr{WMD} and the \abbr{SCM}
      with regularized and non-regularized word embeddings%
    }%
    \label{alg:knn}
    \begingroup
      \baselineskip=10pt
      \hspace*{\algorithmicindent} \textbf{Input}: Training data~$A$, test data~$B$, neighborhood size~$k$ \\
      \hspace*{\algorithmicindent} \textbf{Output}: Class labels for the test data,~$L$
      \begin{algorithmic}[1]
        \ForEach {$b_i\in B$}
          \State For all $a_j\in A$, compute the \abbr{SCM} / \abbr{WMD} between $a_j$ and~$b_i$.
          \State Choose the $k$ nearest neighbors of $b_i$ in $A$.
          \State Assign the majority class of the $k$ nearest neighbors to~$l_i$.
      \EndFor
      \end{algorithmic}
    \endgroup
  \end{algorithm}
\end{minipage}\hfill
\begin{minipage}[b]{0.48\textwidth}
  \begin{table}[H]
    \caption{%
      Statistics of the datasets used for evaluation, where $E[p]$ is the
      average number of unique words in a document, and $|\mathcal Y|$ is the
      number of document classes%
    }%
    \label{table:dataset-statistic}
    \vspace{-0.2cm}
    \begin{center}
    \begin{tabular}{crrrr}
    \toprule
      \begin{tabular}{@{}c@{}}Dataset \\ name\end{tabular} &
        \begin{tabular}{@{}c@{}}Train set \\ size\end{tabular} &
        \begin{tabular}{@{}c@{}}Test set \\ size\end{tabular} &
        \begin{tabular}{@{}c@{}}$E[p]$\end{tabular} &
        \begin{tabular}{@{}c@{}}$|\mathcal Y|$\end{tabular}
        \\ \midrule
      \dataset{BBCSPORT} &    517 &   220 & 181.0 &  5 \\
      \dataset{TWITTER}  &  2,176 &   932 &  13.7 &  3 \\
      \dataset{OHSUMED}  &  3,999 & 5,153 &  89.4 & 10 \\
      \dataset{REUTERS}  &  5,485 & 2,189 &  56.0 &  8 \\
      \dataset{AMAZON}   &  5,600 & 2,400 &  86.3 &  4 \\
      \dataset{20NEWS}   & 11,293 & 7,528 & 145.0 & 20 \\
      \bottomrule
    \end{tabular}
    \end{center}
    \vspace{-0.29cm}
  \end{table}
\end{minipage}
\end{figure}

\begin{figure*}[p!]
\centering
\mbox{%
  \hspace{-0.2cm}%
  \includegraphics{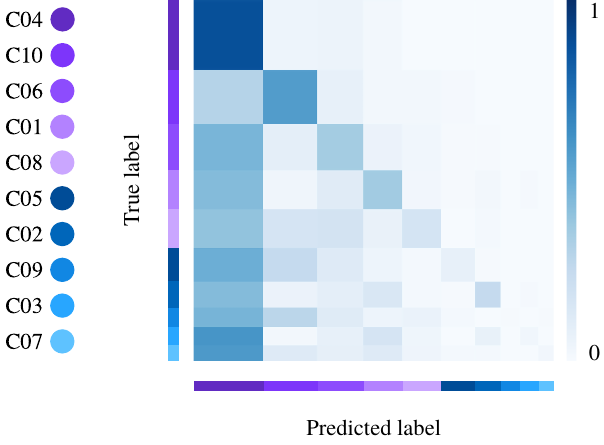}\hspace{0.3cm}%
  \raisebox{0cm}{\includegraphics[width=5cm,height=5cm]{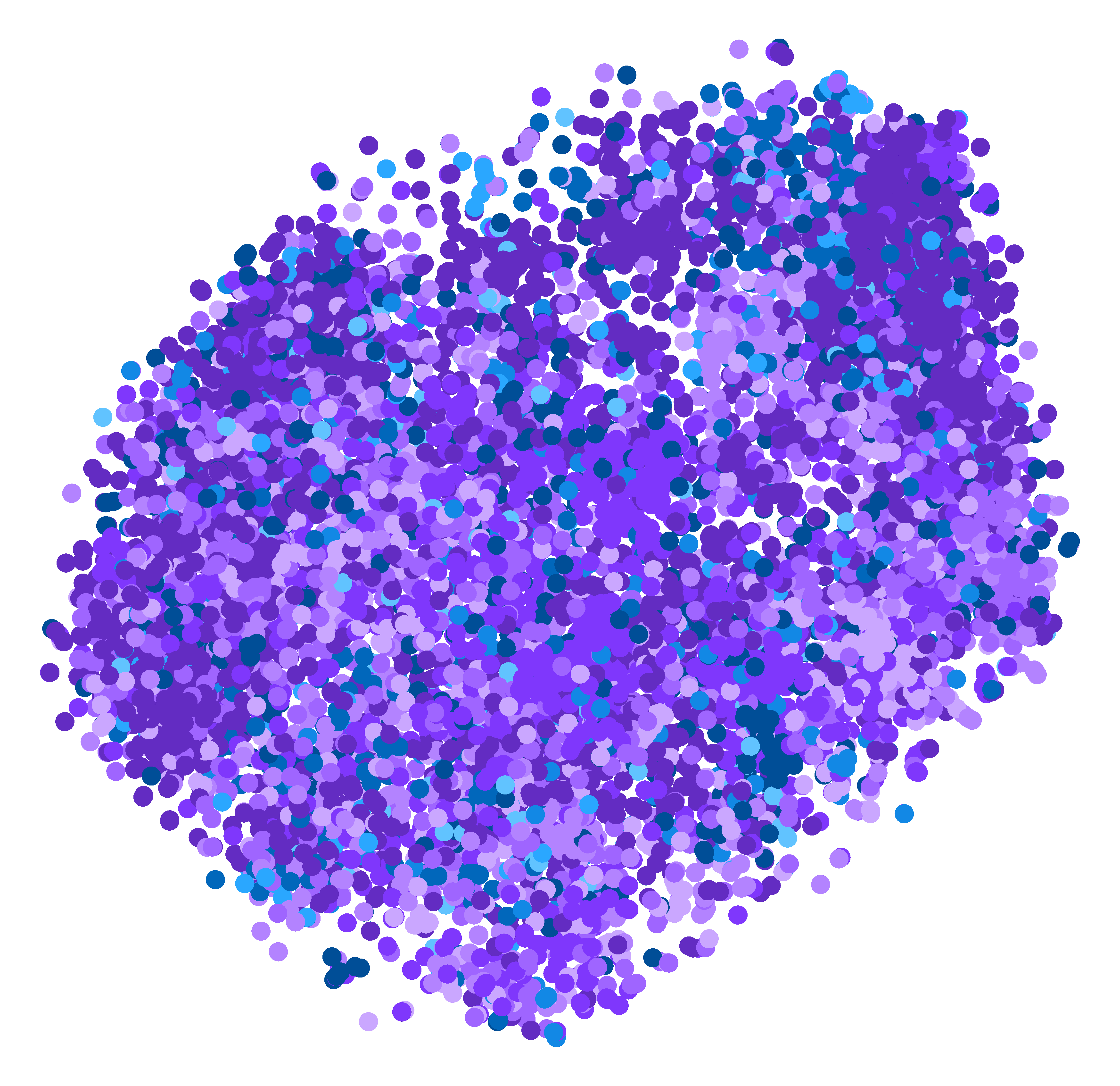}}%
}%
\\[0.5cm]
\mbox{%
  \hspace{-0.2cm}%
  \includegraphics{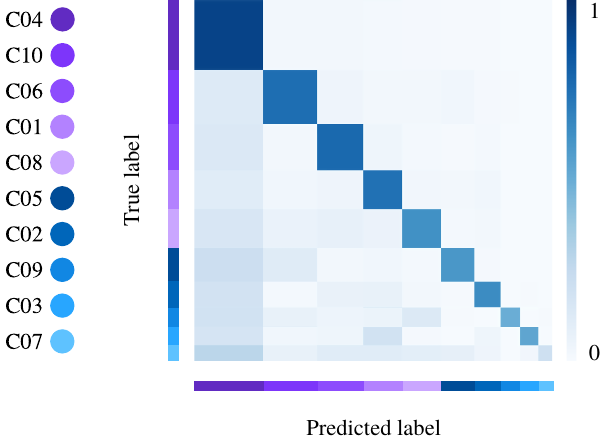}\hspace{0.3cm}%
  \raisebox{0cm}{\includegraphics[width=5cm,height=5cm]{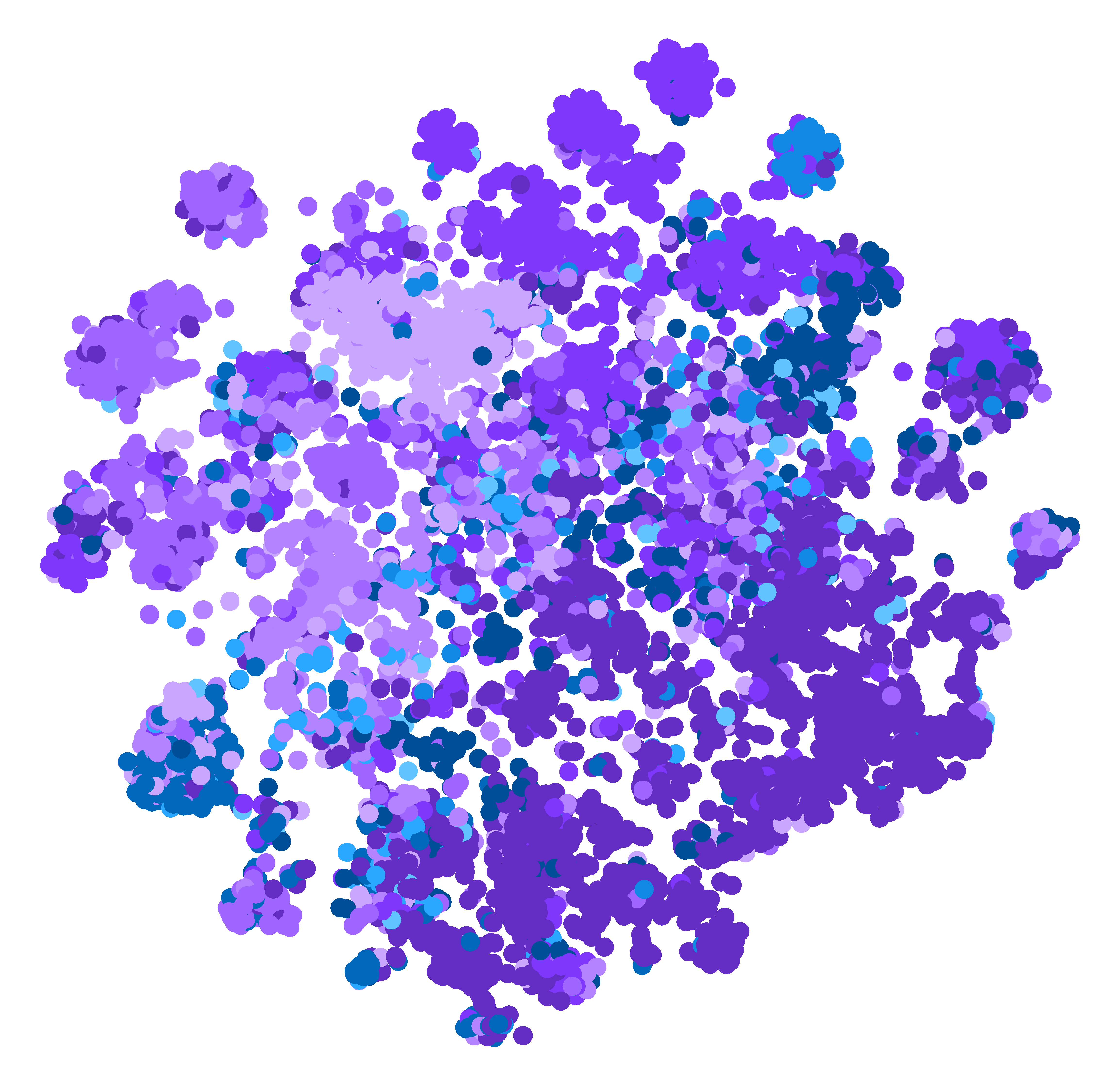}}%
}%
\caption{%
  \kNN{} ($k=9$) confusion matrices and t-\abbr{SNE} document visualizations
  for the soft \abbr{VSM} with non-regularized (50.71\% test error, top) and
  regularized (24.14\% test error, bottom) word embeddings on the
  \dataset{OHSUMED} dataset.
  See \url{https://mir.fi.muni.cz/ohsumed-nonregularized}
  and \url{https://mir.fi.muni.cz/ohsumed-regularized} for interactive three-dimensional
  t-\abbr{SNE} document visualisations.
}
\label{fig:classification-ohsumed}
\end{figure*}

\begin{figure*}[p!]
\centering
\mbox{%
  \hspace{-0.2cm}%
  \includegraphics{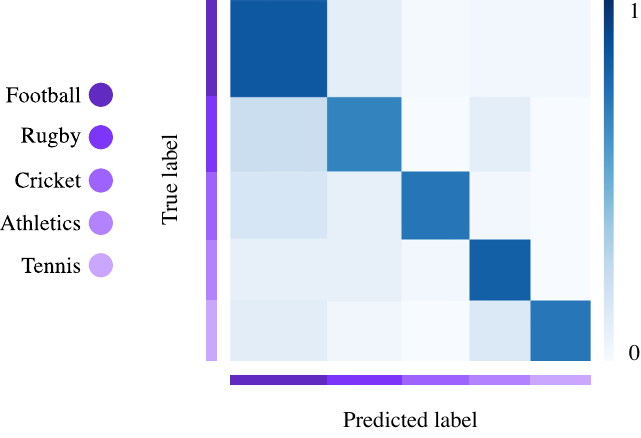}\hspace{0.3cm}%
  \raisebox{0cm}{\includegraphics[width=5cm,height=5cm]{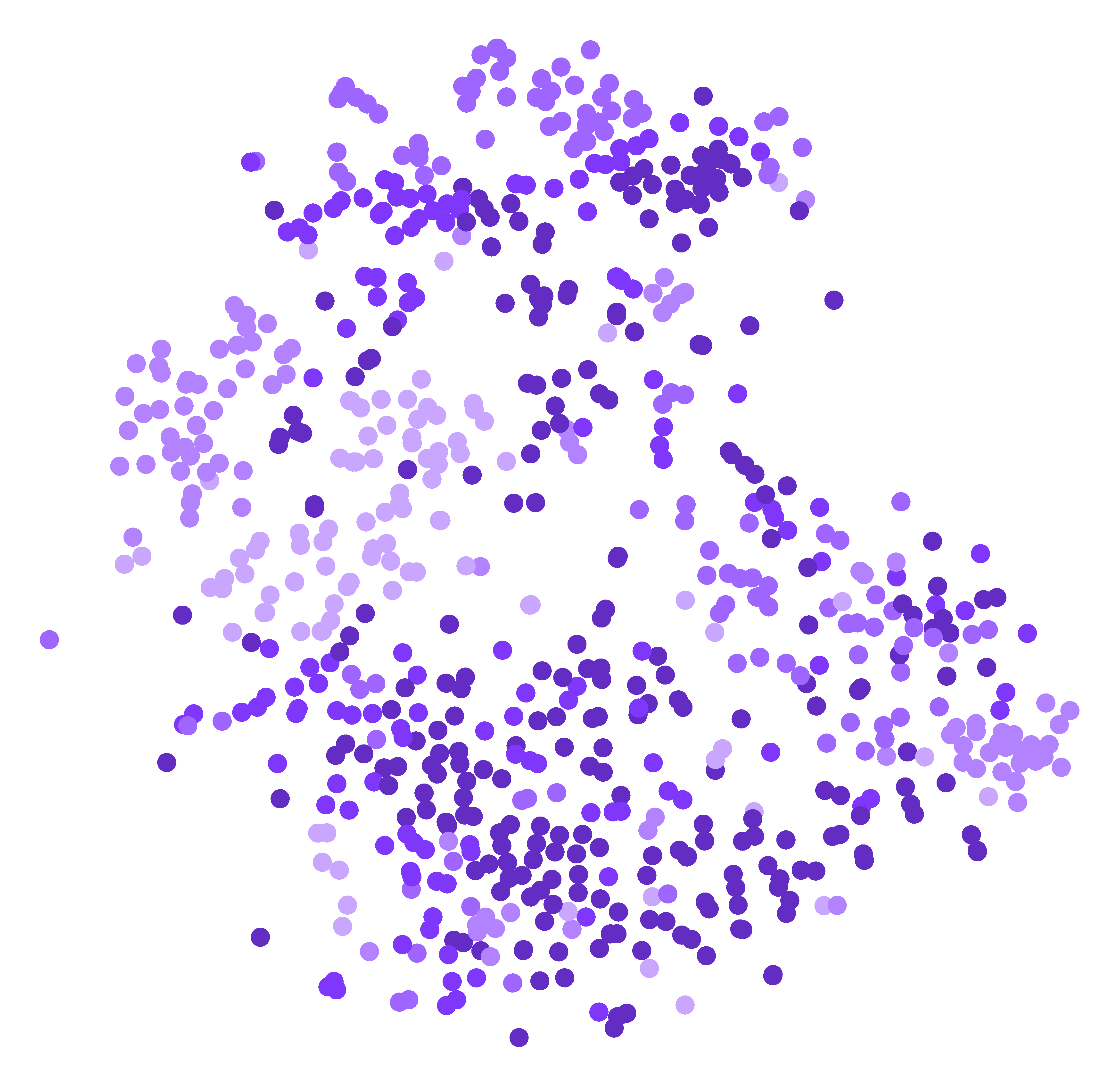}}%
  \hspace{0.4cm}%
}%
\\[0.5cm]
\mbox{%
  \hspace{-0.2cm}%
  \includegraphics{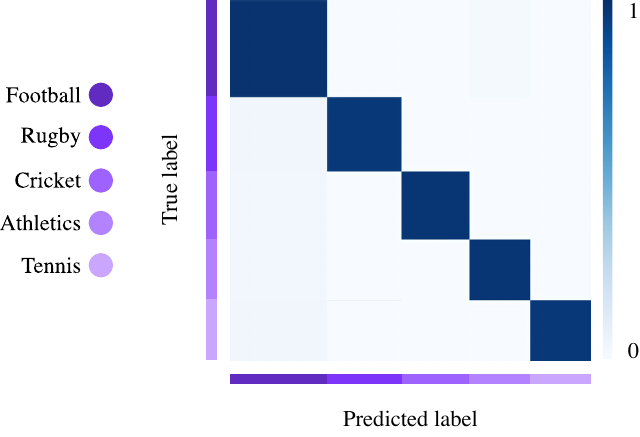}\hspace{0.3cm}%
  \raisebox{0cm}{\includegraphics[width=5cm,height=5cm]{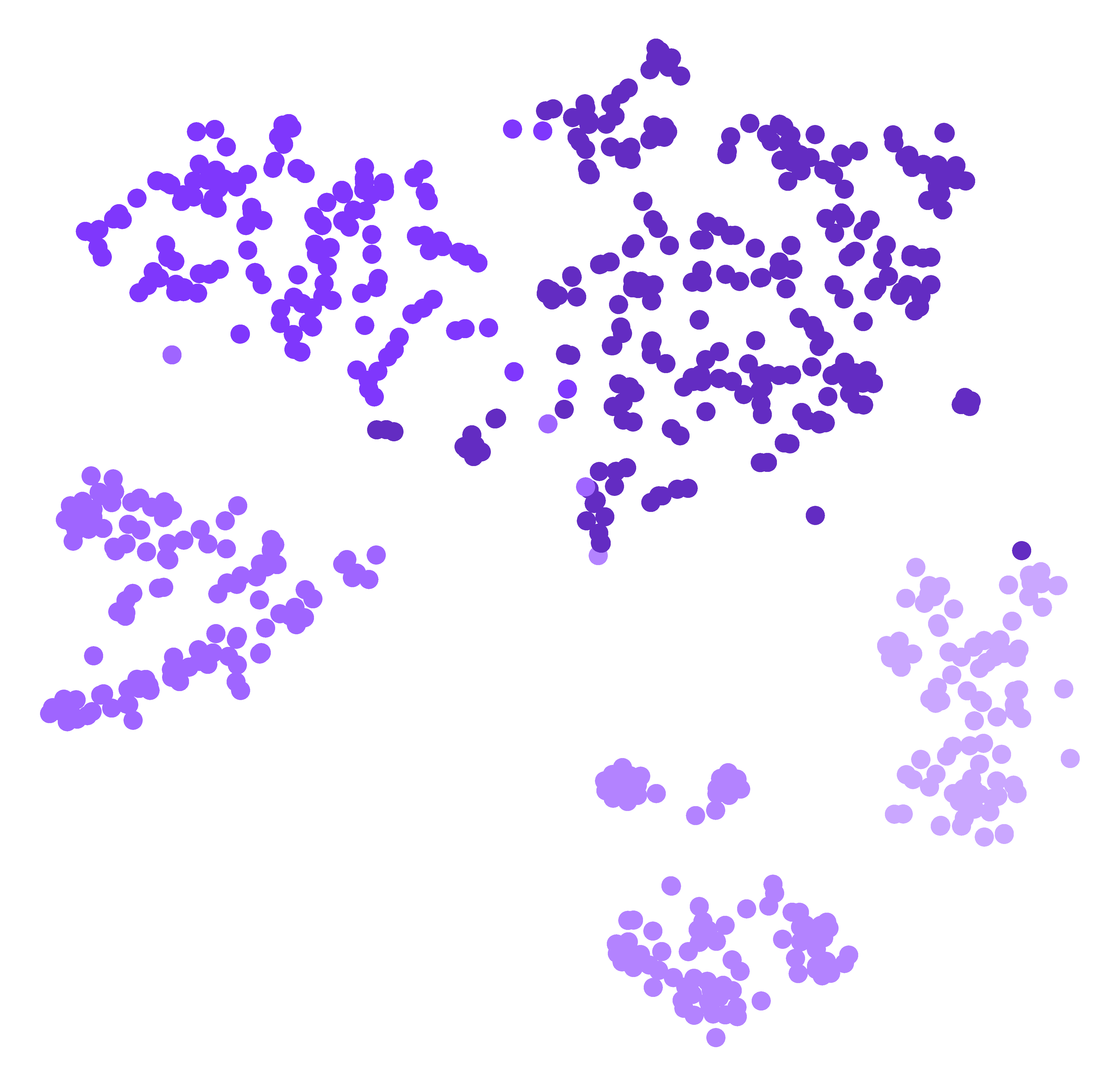}}%
  \hspace{0.4cm}%
}%
\caption{%
  \kNN{} ($k=1$) confusion matrices and t-\abbr{SNE} document visualizations
  for the soft \abbr{VSM} with non-regularized (21.82\% test error, top) and
  regularized (2.27\% test error, bottom) word embeddings on the
  \dataset{BBCSPORT} dataset.
  See \url{https://mir.fi.muni.cz/bbcsport-nonregularized}
  and \url{https://mir.fi.muni.cz/bbcsport-regularized} for interactive three-dimensional
  t-\abbr{SNE} document visualisations.
}
\label{fig:classification-bbcsport}
\end{figure*}

\iffalse
\begin{figure*}
\centering
\mbox{%
  \hspace{-0.2cm}%
  \includegraphics{classification-confusion_matrices-twitter-dense}\hspace{0.3cm}%
  \raisebox{0cm}{\includegraphics[width=5cm,height=5cm]{classification-datasets-twitter-dense}}%
}%
\\[0.5cm]
\mbox{%
  \hspace{-0.2cm}%
  \includegraphics{classification-confusion_matrices-twitter-sparse}\hspace{0.3cm}%
  \raisebox{0cm}{\includegraphics[width=5cm,height=5cm]{classification-datasets-twitter-sparse}}%
}%
\caption{%
  \kNN{} ($k=17$ and $11$) confusion matrices and t-\abbr{SNE} document
  visualizations for the soft \abbr{VSM} with non-regularized (31.22\% test
  error, top) and regularized (31.01\% test error, bottom) word embeddings on
  the \dataset{TWITTER} dataset.
  See \url{https://mir.fi.muni.cz/twitter-nonregularized}
  and \url{https://mir.fi.muni.cz/twitter-regularized} for interactive three-dimensional
  t-\abbr{SNE} document visualisations.
}
\label{fig:classification-twitter}
\end{figure*}
\fi

\begin{figure*}
\centering
\mbox{%
  \hspace{-0.2cm}%
  \includegraphics{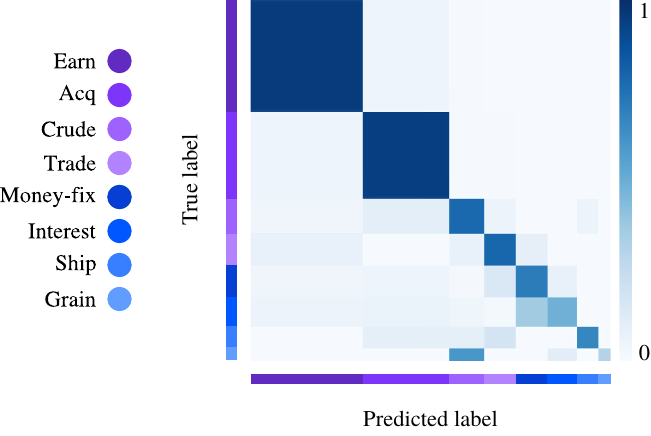}\hspace{0.3cm}%
  \raisebox{0cm}{\includegraphics[width=5cm,height=5cm]{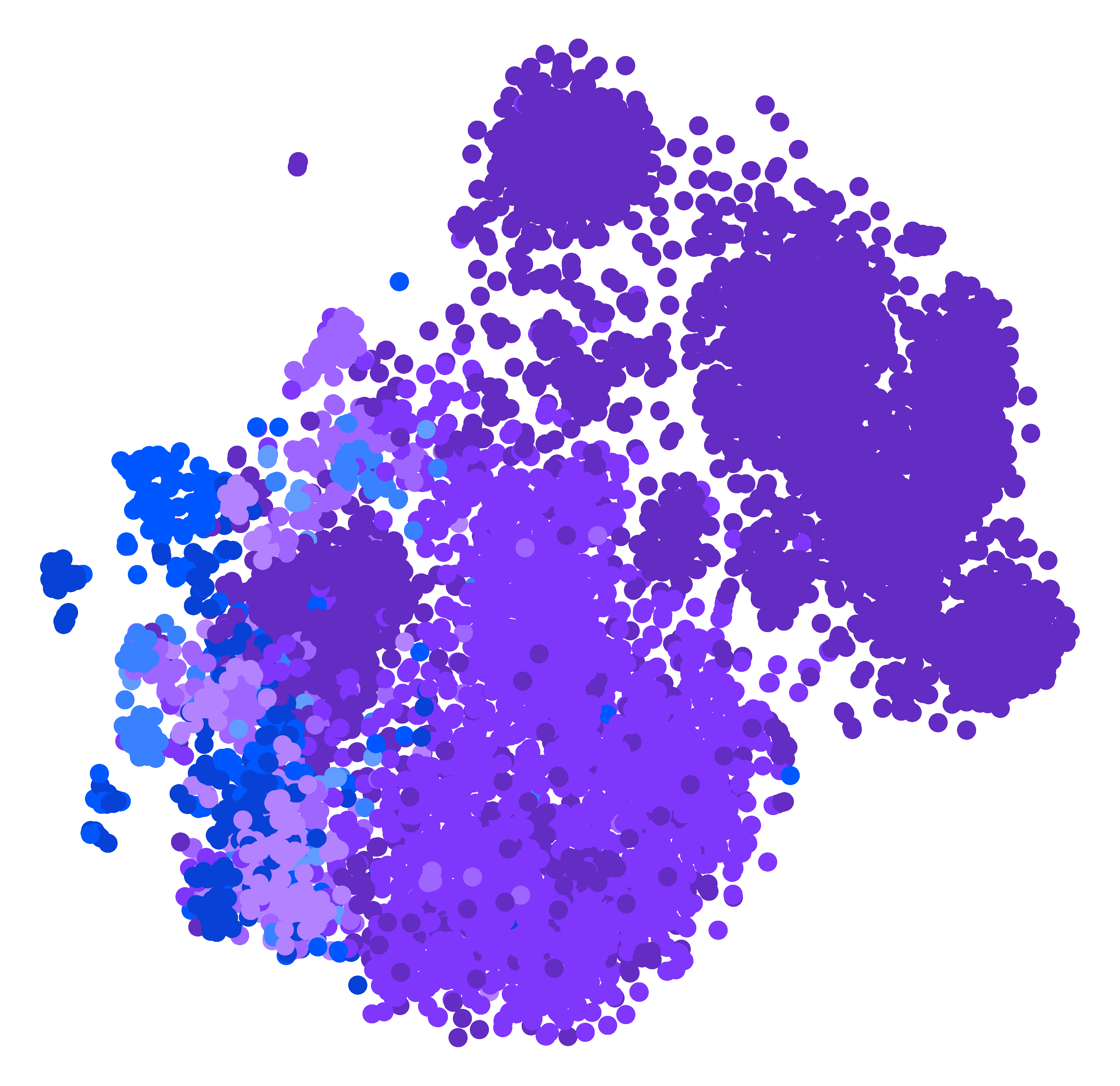}}%
}%
\\[0.5cm]
\mbox{%
  \hspace{-0.2cm}%
  \includegraphics{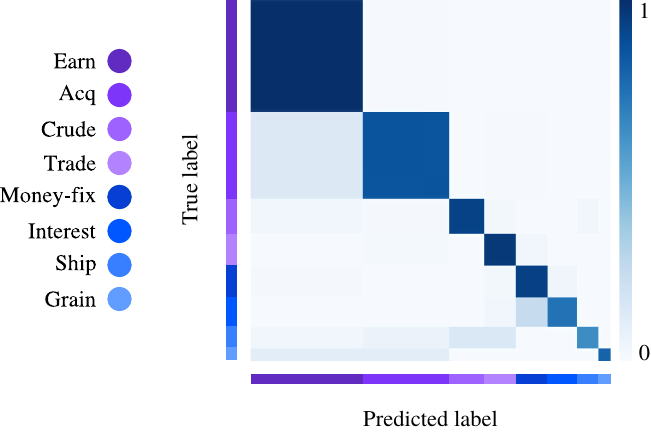}\hspace{0.3cm}%
  \raisebox{0cm}{\includegraphics[width=5cm,height=5cm]{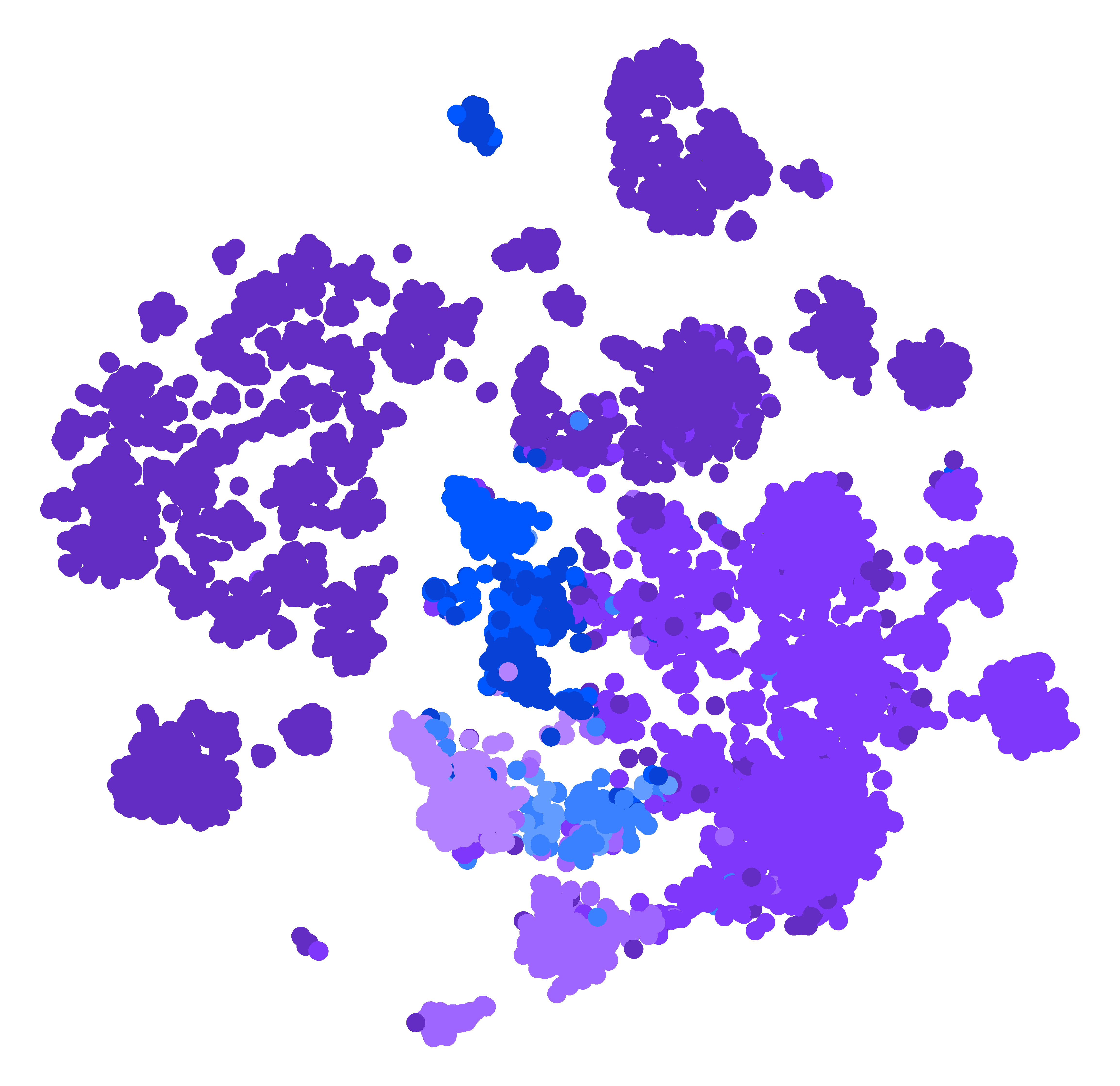}}%
}%
\caption{%
  \textls[-20]{%
    \kNN{} ($k=9$ and $19$) confusion matrices and t-\abbr{SNE} document
    visualizations for the soft \abbr{VSM} with non-regularized%
  }
  (10.19\% test error, top) and regularized (7.22\% test error, bottom) word
  embeddings on the \dataset{REUTERS} dataset.
  See \url{https://mir.fi.muni.cz/reuters-nonregularized}
  and \url{https://mir.fi.muni.cz/reuters-regularized} for interactive three-dimensional
  t-\abbr{SNE} document visualisations.
}
\label{fig:classification-reuters}
\end{figure*}

\begin{figure*}
\centering
\mbox{%
  \hspace{-0.2cm}%
  \includegraphics{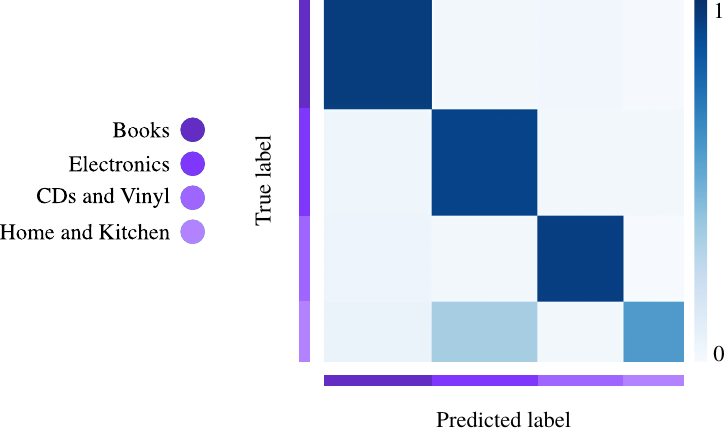}\hspace{0.3cm}%
  \raisebox{0cm}{\includegraphics[width=5cm,height=5cm]{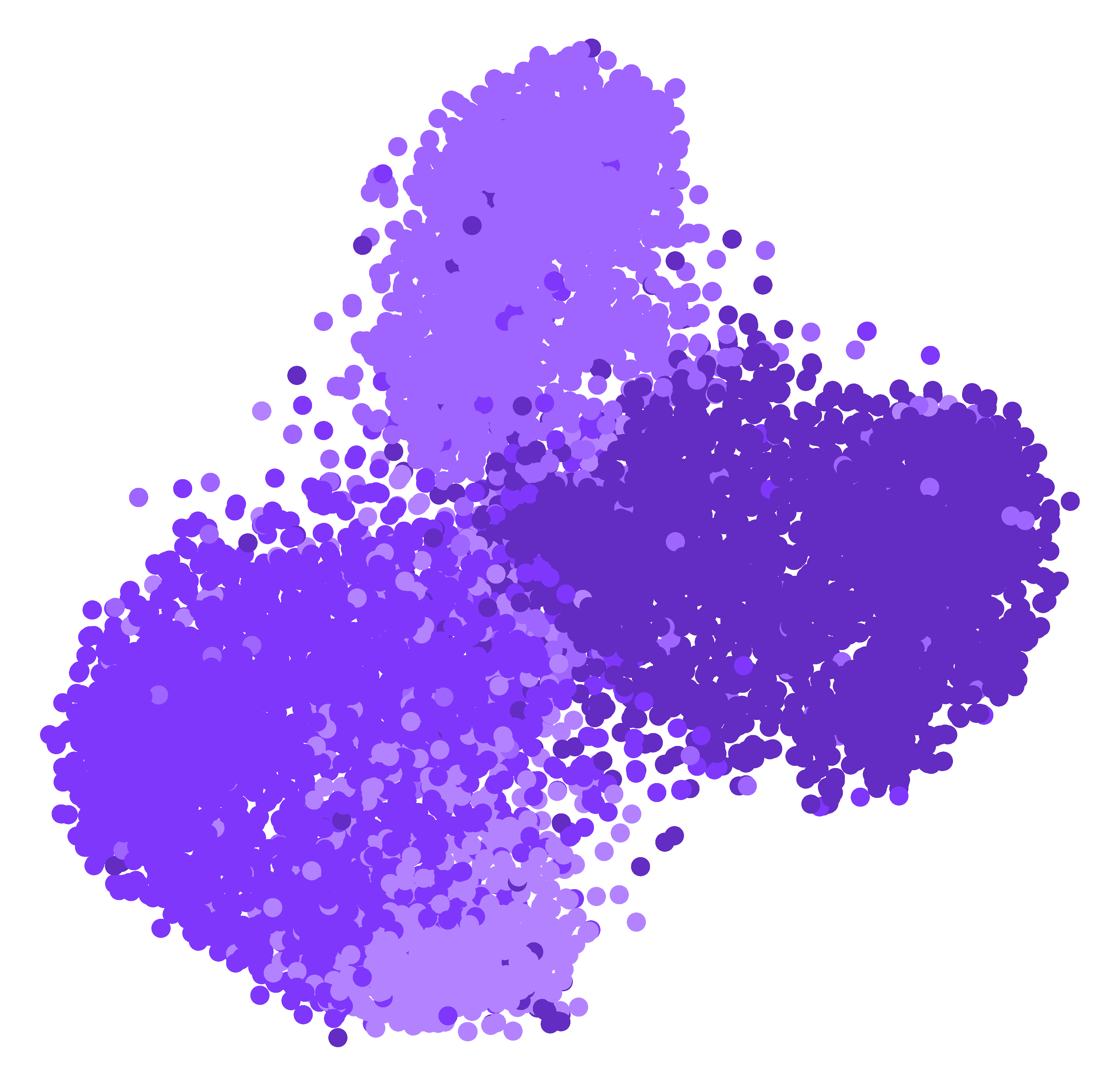}}%
  \hspace{0.75cm}%
}%
\\[0.5cm]
\mbox{%
  \hspace{-0.2cm}%
  \includegraphics{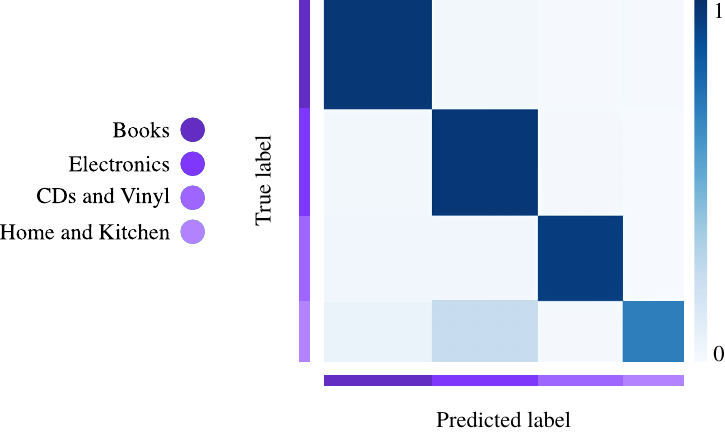}\hspace{0.3cm}%
  \raisebox{0cm}{\includegraphics[width=5cm,height=5cm]{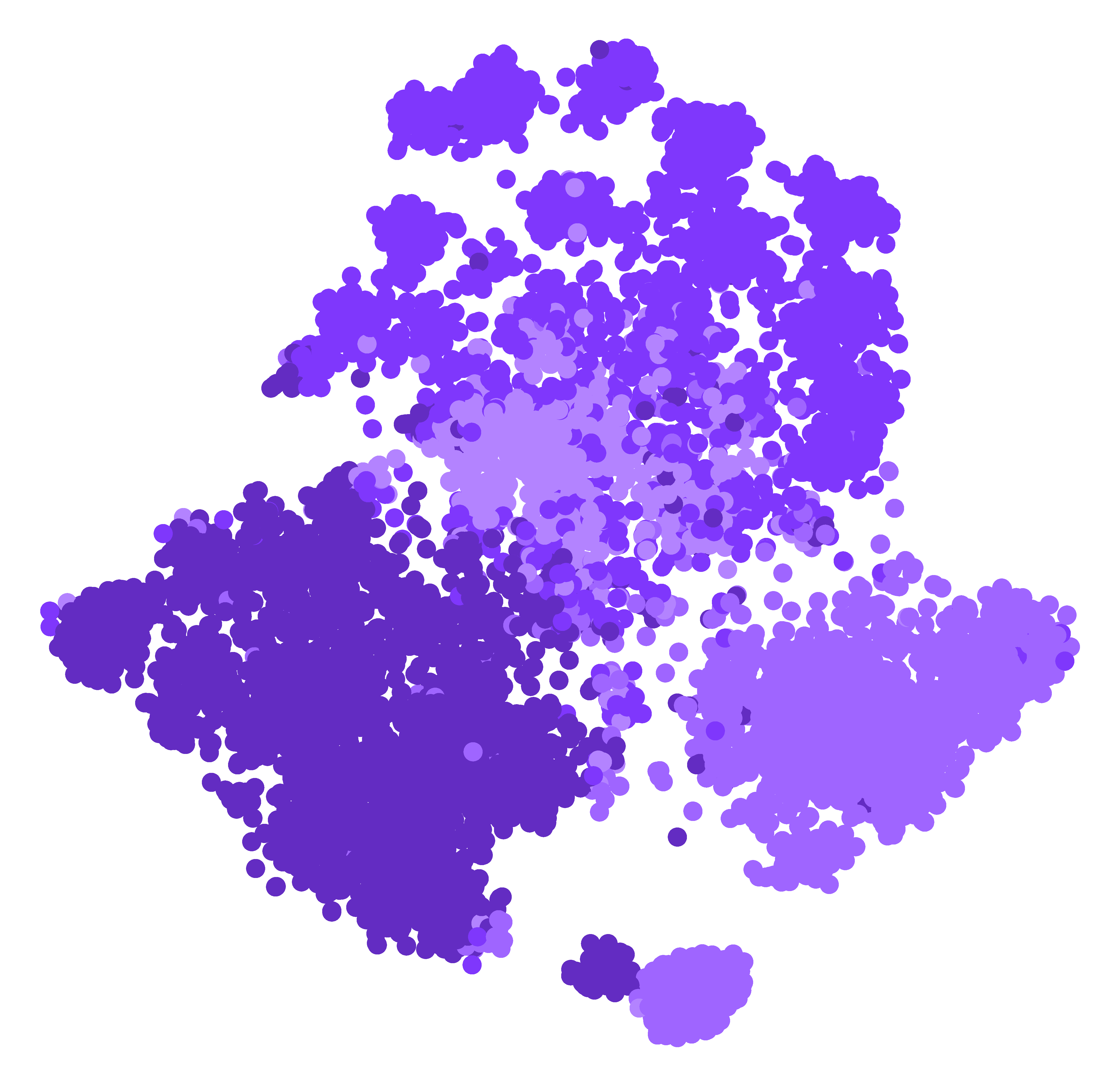}}%
  \hspace{0.75cm}%
}%
\caption{
  \textls[-20]{%
    \kNN{} ($k=15$ and $11$) confusion matrices and t-\abbr{SNE} document
    visualizations for the soft \abbr{VSM} with non-regularized%
  }
  (10.04\% test error, top) and regularized (6.33\% test error, bottom) word
  embeddings on the \dataset{AMAZON} dataset.
  See \url{https://mir.fi.muni.cz/amazon-nonregularized}
  and \url{https://mir.fi.muni.cz/amazon-regularized} for interactive three-dimensional
  t-\abbr{SNE} document visualisations.
}
\label{fig:classification-amazon}
\end{figure*}

\begin{figure*}
\vspace{0.5cm}
\includegraphics{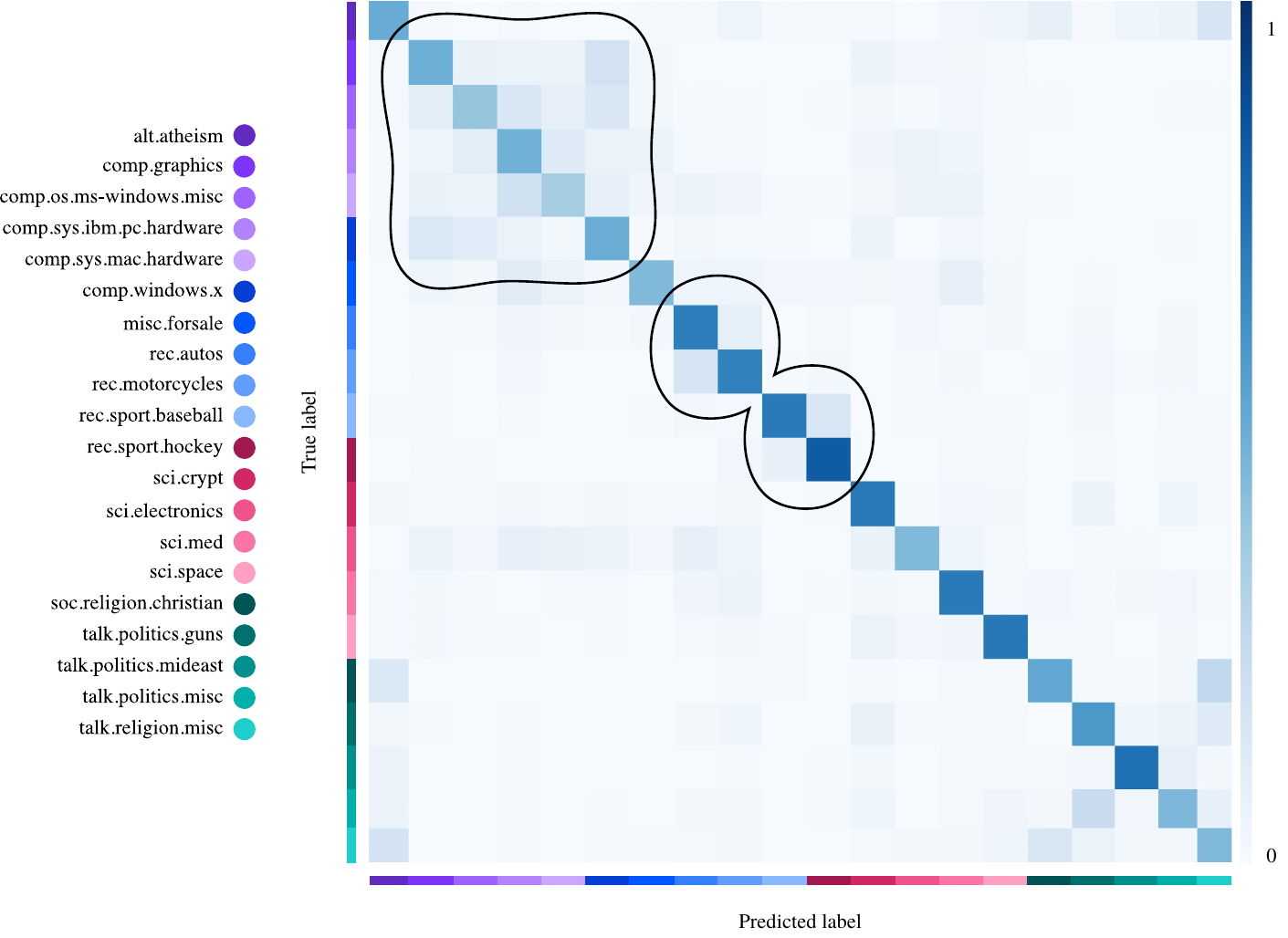}%
\\[0.42cm]
\hspace*{3.3cm}%
\includegraphics[width=11cm,height=11cm]{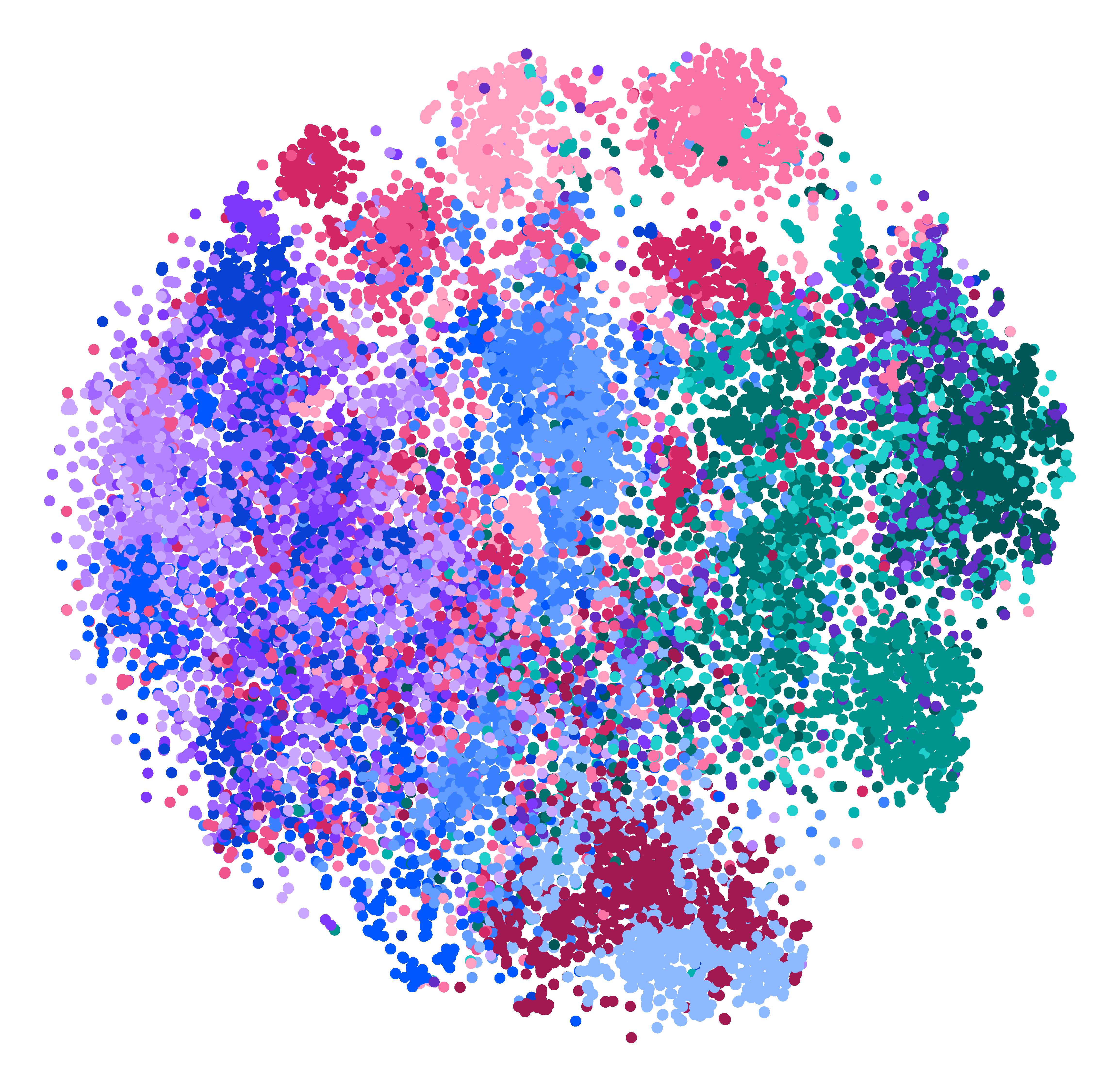}%
\kern-11cm%
\includegraphics[width=11cm,height=11cm]{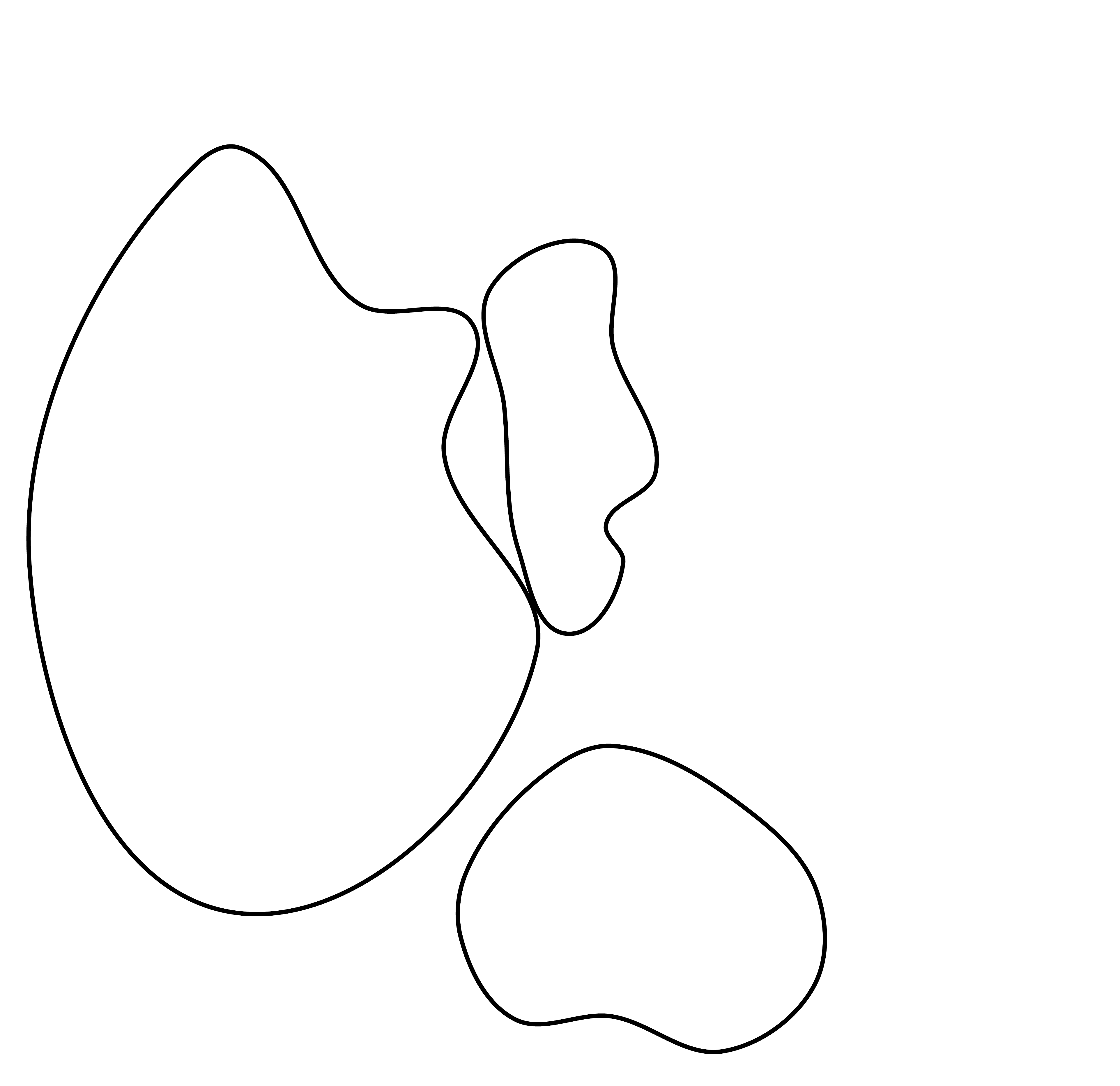}%
\caption{%
  \kNN{} ($k=1$) confusion matrix (top) and t-\abbr{SNE} document visualization
  (bottom) for the soft \abbr{VSM} with non-regularized (42.53\% test error)
  word embeddings on the \dataset{20NEWS} dataset.
  Notice how the Usenet hierarchies \texttt{comp.*}, \texttt{rec.*},
  and \texttt{rec.sport.*} form visible clusters.
  See \url{https://mir.fi.muni.cz/20news-nonregularized} for an interactive three-dimensional
  t-\abbr{SNE} document visualisation.
}
\label{fig:classification-20news-a}
\end{figure*}

\begin{figure*}
\vspace{0.5cm}
\includegraphics{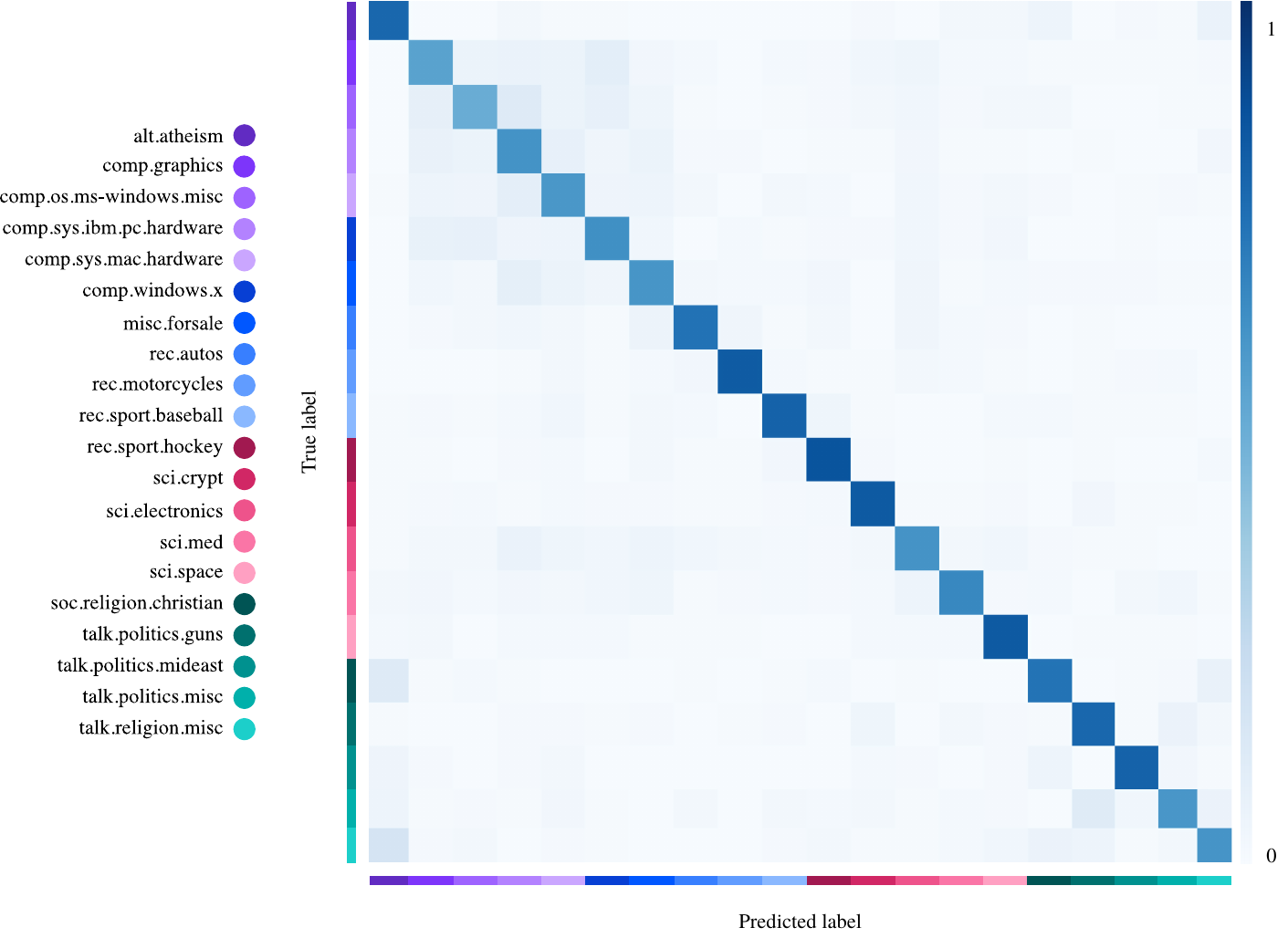}%
\\[0.42cm]
\hspace*{3.3cm}\includegraphics[width=11cm,height=11cm]{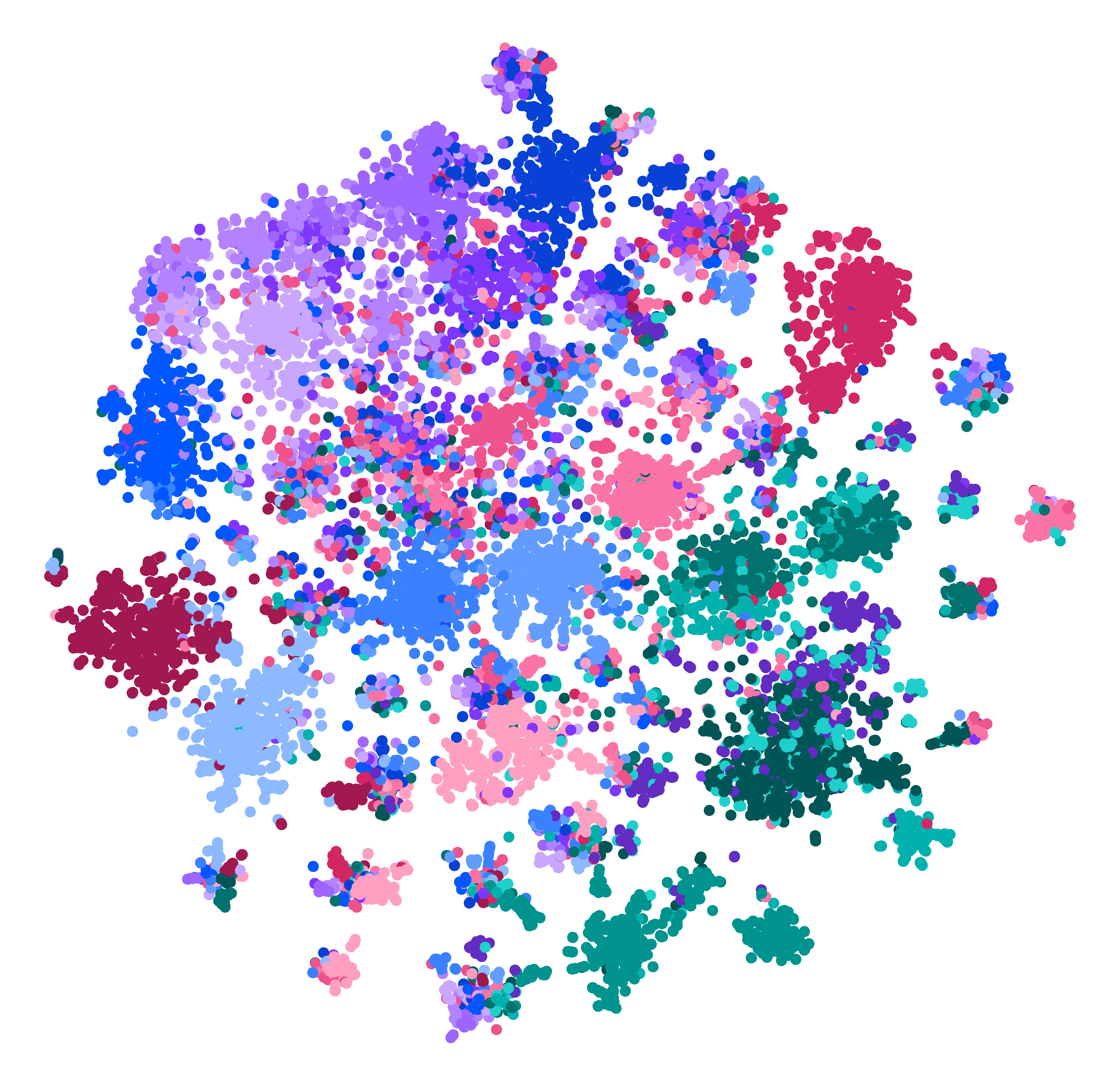}
\caption{%
  \kNN{} ($k=1$) confusion matrix (top) and t-\abbr{SNE} document visualization
  (bottom) for the soft \abbr{VSM} with regularized (29.97\% test error) word
  embeddings on the \dataset{20NEWS} dataset.
  See \url{https://mir.fi.muni.cz/20news-regularized} for an interactive three-dimensional
  t-\abbr{SNE} document visualisation.
}
\label{fig:classification-20news-b}
\end{figure*}

\subsection{Training and Regularization of Word Embeddings}
Using Word2Vec, we train the \abbr{CBOW} on the first 100 MiB of the English
Wikipedia~\citep{mahoney11about} using the same
parameters as \citet[Section~4.3]{lam} and 10 training epochs. We use quantized
word embeddings in 1,000 dimensions and non-quantized word embeddings
in 200 dimensions to achieve comparable performance on the word analogy
task.~\citep{lam}

\subsection{Nearest Neighbor Classification}
We use the \abbr{VSM} with uniform word frequency weighting, also known as the
bag of words (\abbr{BoW}), as our baseline. For the \abbr{SCM}, we use the
double-logarithm inverse collection frequency word weighting (the \abbr{SMART}
\texttt{dtb} weighting scheme) as suggested by \citet[Table~1]{singhal2001modern}.
For the \abbr{WMD}, we use \abbr{BoW} document document vectors like \citet{kusner2015word}.

We tune the parameters $o\in\{1,2,3,4\}$ and $t\in\{0, \pm\nicefrac{1}{2}, 1\}$
of the \abbr{SCM}, the parameter
$s\in\{0.0, 0.1, \ldots, 1.0\}$ of the \abbr{SMART} \texttt{dtb} weighting
scheme, the parameter $k\in\{1,3,\ldots,19\}$ of the \kNN, and the parameters
$C\in\{100, 200, 400, 800\}$, and Idf, Sym, Dom${}\in\{\true,\false\}$ of the
orthogonalization. For each dataset, we hold out 20\% of the train set for
validation, and we use grid search to find the optimal parameter values.
To classify each sample in the test set, we follow the procedure presented in
Algorithm~\ref{alg:knn}.

\subsection{Significance Testing}
We use the method of \citet{agresti1998approximate} to construct 95\%
confidence intervals for the \kNN{} test error. For every dataset, we use
Student's $t$-test at 95\% confidence level with
$q$-values~\citep{benjamini1995controlling} for all combinations of document
similarities and word embedding regularization techniques to find significant
differences in \kNN{} test error.

\section{Discussion of Results}
\label{sec:results}
Our results are shown in figures~\ref{fig:classification-ohsumed}--\ref{fig:similarities}
and in Table~\ref{table:optimal-parameters}. In the following subsections, we will
discuss the individual results and how they are related.

\subsection{Task Performance on Individual Datasets}
Figure~\ref{fig:test-error} shows 95\% interval estimates for the \kNN{}
test error. All differences are significant, except for the second and
fourth results from the left on the \dataset{BBCSPORT} and \dataset{TWITTER}
datasets, the sixth and seventh results from the left on the
\dataset{BBCSPORT} and \dataset{OHSUMED} datasets, and the fourth and sixth
results from the left on the \dataset{TWITTER} dataset.
\looseness=-1

\begin{figure*}
\vspace{-0.2cm}%
\begin{minipage}[t]{\textwidth}
  \begin{figure}[H]
    \hspace*{-0.4cm}%
    \includegraphics{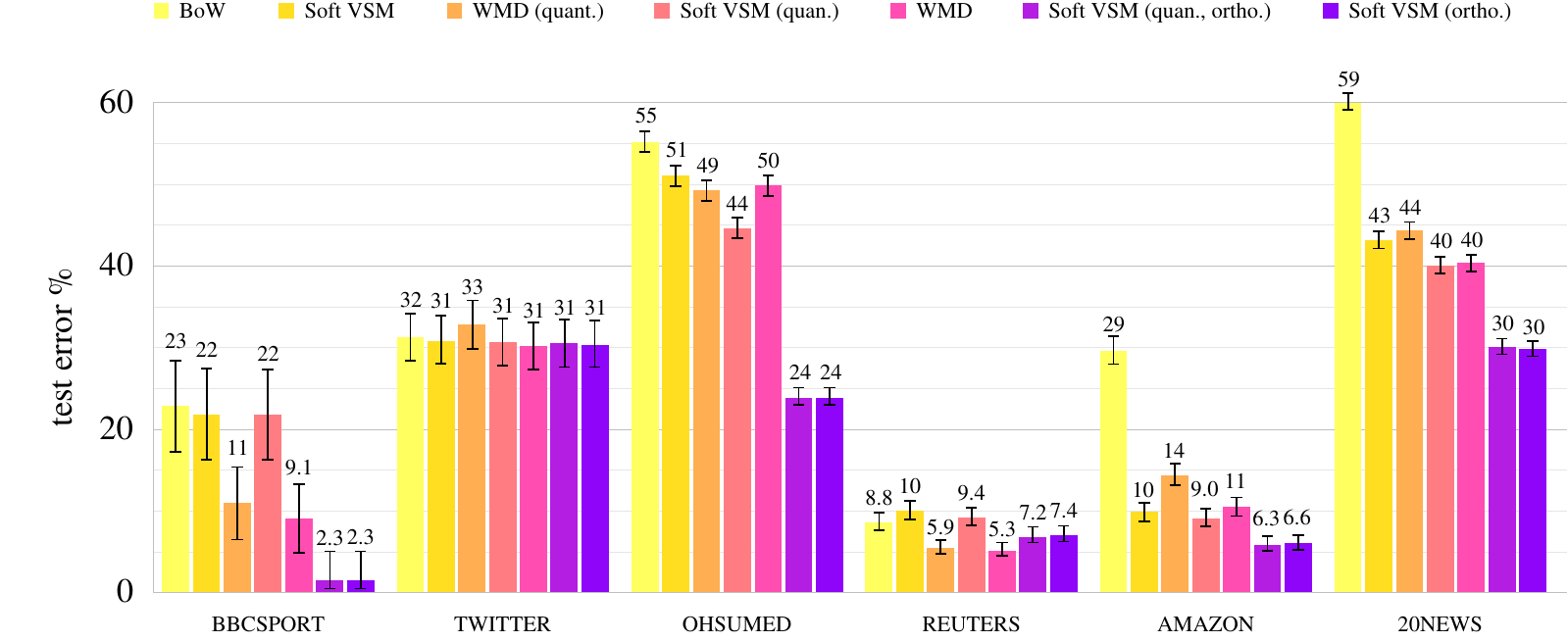}
    \caption{%
      95\% interval estimates for the \kNN{} test error on six text classification datasets%
    }%
    \label{fig:test-error}
  \end{figure}
\end{minipage}\\[-0.2cm]
\begin{minipage}[b]{0.48\textwidth}
  \begin{figure}[H]
    \centering
    \hspace*{0.06cm}%
    \includegraphics[scale=0.98]{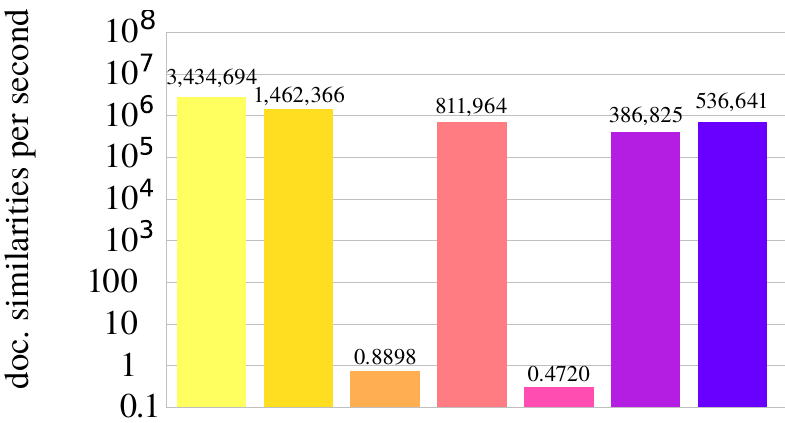}
    \caption{%
      \textls[-35]{%
        Average document processing speed on one Intel Xeon X7560 core%
      }
    }
    \label{fig:speed}
  \end{figure}
\end{minipage}\hfill
\begin{minipage}[b]{0.452\textwidth}
  \begin{figure}[H]
    \centering
    \includegraphics[scale=0.98]{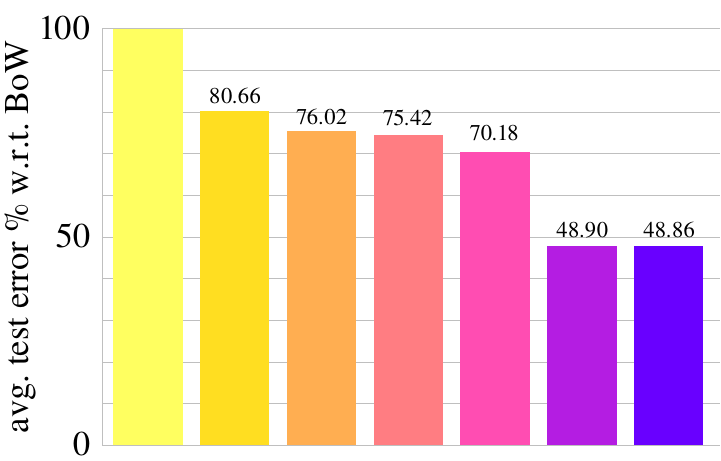}
    \caption{%
      Average \kNN{} test error compared to the \abbr{BOW}%
    }
    \label{fig:avg-test-error-ratio}
  \end{figure}
\end{minipage}
\end{figure*}

Although most differences in the \kNN{} test error are statistically significant on the
\dataset{TWITTER} dataset, they are relatively minor compared to other
datasets. This is because of two reasons: (1)~\dataset{TWITTER} is a sentiment
analysis dataset, and (2)~words with opposite sentiment often appear in similar
sentences. Since word embeddings are similar for words that appear in similar
sentences, embeddings for words with opposite sentiment are often similar. For
example, the embeddings for the words good and bad are mutual nearest neighbors
with cosine similarity 0.58 for non-quantized and 0.4 for quantized word
embeddings. As a result, word embeddings don't help us separate positive and
negative documents. Using a better measure of sentiment in the word similarity
matrix~$\S$ and in the flow cost $c_{ij}$ would improve the task
performance of the soft \abbr{VSM} and the~\abbr{WMD}.

\textls[-15]{%
Figures~\ref{fig:classification-ohsumed}--\ref{fig:classification-20news-b}
show confusion matrices and t-\abbr{SNE} document visualizations~\citep{maaten2008visualizing}
for the soft \abbr{VSM} with non-regularized word embeddings and for the
soft \abbr{VSM} with orthogonalized and quantized word embeddings.%
\looseness=-1
}

In Figure~\ref{fig:classification-ohsumed}, we see that with non-regularized word
embeddings, we predict most documents as class~C04, followed by classes~C10 and
C06. When a document representation is uniformly random, then the most common
classes in a dataset are most likely to be predicted by the \kNN{} classifier.
In the \dataset{OHSUMED} dataset, 2,835 documents belong to the most common
class~C04, 1,711 documents belong to the second most common class~C10, and
1,246 documents belong to the third most common class~C06. In contrast to the
almost random representations of the soft \abbr{VSM} with non-regularized word
embeddings, orthogonalized word embeddings make the true class label much easier
to predict. We hypothesize that this is because \dataset{OHSUMED} is a medical
dataset. Medical terms are highly specific, so when we search for documents
containing similar words, we need to restrict ourselves only to highly
similar words, which is one of the effects of using orthogonalized word embeddings.

In Figure~\ref{fig:classification-reuters}, we see that with non-regularized word
embeddings, most documents from the class Grain are misclassified as the more
common class Crude. With regularized word embeddings, the classes are separated.
We hypothesize that this is because both grain and crude oil are commodities,
which makes the documents from both classes contain many similar words. The classes
will become indistinguishable unless we restrict ourselves only to a subset of the
similar words, which is one of the effects of using orthogonalized word embeddings.

In Figure~\ref{fig:classification-20news-a}, we see that with non-regularized word
embeddings, messages are often misclassified as a different newsgroup in the same
Usenet hierarchy. We can see that the Usenet hierarchies \texttt{comp.*}, \texttt{rec.*},
and \texttt{rec.sport.*} form visible clusters in both the t-\abbr{SNE} document visualization
and in the confusion matrix. In Figure~\ref{fig:classification-20news-b}, we see that
with regularized word embeddings, the clusters are separated. We hypothesize that this
is because newsgroups in a Usenet hierarchy share a common topic and similar terminology.
The terminology of the newsgroups will become difficult to separate unless only highly
specific words are allowed to be considered similar, which is one of the effects of using
orthogonalized word embeddings.

\subsection{Average Task Performance}

Figure~\ref{fig:avg-test-error-ratio} shows the average \kNN{} test error ratio
between the document similarities and the \abbr{BOW} baseline. This ratio is
the lowest for the soft \abbr{VSM} with orthogonalized word embeddings, which
achieves $48.86\%$ of the average \abbr{BOW} \kNN{} test error.  The average
test error ratio between the soft \abbr{VSM} with regularized word embeddings
and the soft \abbr{VSM} with non-regularized word embeddings is $60.57\%$,
which is a $39.43\%$ reduction in the average \kNN{} test error.
\looseness=-1

Unlike the soft \abbr{VSM}, the \abbr{WMD} does not benefit from word embedding
quantization. This is because of two reasons: (1)~the soft \abbr{VSM} takes into
account the similarity between all words in two documents, whereas the
\abbr{WMD} only considers the most similar word pairs, and (2)~non-quantized
word embeddings are biased towards positive similarity, see
Figure~\ref{fig:similarities}. With non-quantized word embeddings, embeddings
of unrelated words have positive cosine similarity, which makes dissimilar
documents less separable. With quantized embeddings, unrelated words have
negative cosine similarity, which improves separability and reduces \kNN{} test error.
The \abbr{WMD} is unaffected by the bias in non-quantized word embeddings, and
the reduced precision of quantized word embeddings increases \kNN{} test error.
\looseness=-1

\subsection{Document Processing Speed}

Figure~\ref{fig:speed} shows the average document processing speed using a
single Intel Xeon X7560 2.26\,GHz core.  Although the orthogonalization reduces
the worst-case time complexity of the soft \abbr{VSM} from quadratic to linear,
it also makes the word similarity matrix sparse, and performing sparse instead
of dense matrix operations causes a $2.73\times$ slowdown compared to the soft
\abbr{VSM} with non-orthogonalized word embeddings. Quantization causes a
$1.8\times$ slowdown, which is due to the $5\times$ increase in the word
embedding dimensionality, since we use 1000-dimensional quantized word
embeddings and only 200-dimensional non-quantized word embeddings.

The super-cubic average time complexity of the \abbr{WMD} results in an average
$819{,}496\times$ slowdown compared to the soft \abbr{VSM} with orthogonalized
and quantized word embeddings, and a $1{,}505{,}386\times$ slowdown on the
\dataset{20NEWS} dataset, which has a large average number of unique words in a
document. Although \citet{kusner2015word} report up to $110\times$ speed-up
using an approximate variant of the \abbr{WMD} (the \abbr{WCD}), this still
results in an average $7{,}450\times$ slowdown, and a $13{,}685\times$ slowdown
on the \dataset{20NEWS} dataset.

\begin{figure}
\vspace{-0.4cm}
\begin{minipage}[c]{0.46\textwidth}
  \vspace{-0.5cm}
  \begin{figure}[H]
    \centering
    \hspace*{-7mm}
    \includegraphics[scale=0.95]{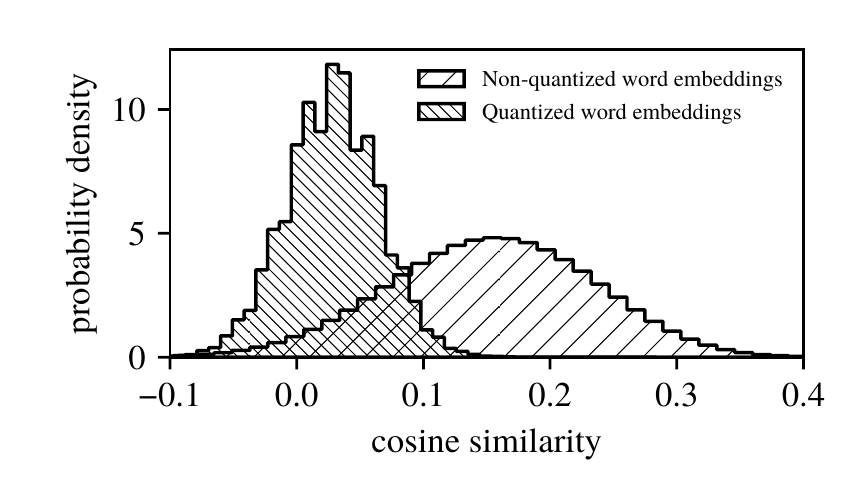}
    \vspace{-0.8cm}
    \caption{%
      Histograms of word embedding cosine similarity%
    }%
    \label{fig:similarities}
  \end{figure}
\end{minipage}\hfill
\begin{minipage}[c]{0.50\textwidth}
  \begin{table}[H]
    \caption{%
      Optimal parameter values for the soft \abbr{VSM} with orthogonalized
      word embeddings. Most common parameter values are bold.
    }%
    \label{table:optimal-parameters}
    \vspace{-0.2cm}
    \begin{center}
    \begin{tabular}{cccccrccc}
    \toprule
      \begin{tabular}{@{}c@{}}Dataset \\ name\end{tabular} &
        \begin{tabular}{@{}c@{}}$o$\end{tabular} &
        \begin{tabular}{@{}c@{}}$s$\end{tabular} &
        \begin{tabular}{@{}c@{}}$t$\end{tabular} &
        \begin{tabular}{@{}c@{}}$C$\end{tabular} &
        \begin{tabular}{@{}c@{}}$k$\end{tabular} &
        \begin{tabular}{@{}c@{}}\kern0.5ex Idf\kern0.5ex\end{tabular} &
        \begin{tabular}{@{}c@{}}Sym\end{tabular} &
        \begin{tabular}{@{}c@{}}Dom\end{tabular}
        \\ \midrule
      \dataset{BBCSPORT} &         2 & \textbf 0 & $\mathbf{-1}$ & \textbf{100} & \textbf 1 & \bftrue & \bftrue & \bftrue \\
      \dataset{TWITTER}  &         2 & \textbf 0 & $\mathbf{-1}$ &         400  &        13 & \bftrue & \bftrue & \false  \\
      \dataset{OHSUMED}  & \textbf 4 & \textbf 0 & $\mathbf{-1}$ &         200  &        11 & \bftrue & \false  & \bftrue \\
      \dataset{REUTERS}  & \textbf 4 & \textbf 0 & $\mathbf{-1}$ & \textbf{100} &        19 & \bftrue & \bftrue & \bftrue \\
      \dataset{AMAZON}   & \textbf 4 & \textbf 0 & $\mathbf{-1}$ & \textbf{100} &        17 & \bftrue & \bftrue & \bftrue \\
      \dataset{20NEWS}   &         3 & \textbf 0 & $\mathbf{-1}$ & \textbf{100} & \textbf 1 & \bftrue & \false  & \false  \\
      \bottomrule
    \end{tabular}
    \end{center}
  \end{table}
\end{minipage}
\end{figure}

\subsection{Parameter Optimization}

Table~\ref{table:optimal-parameters} shows the optimal parameter values for the
soft \abbr{VSM} with orthogonalized word embeddings. The most common parameter
value $\text{Idf}=\true$ shows that it is important to store the nearest
neighbors of rare words in the word similarity matrix $\S'$. The most common
parameter values $C=100$, $o=4$, $t=-1$, $\text{Sym}=\true$, and $\text{Dom}=\true$
show that strong orthogonalization, which makes most values in $\S'$ zero or
close to zero, gives the best results.

Due to the low document processing speed of the \abbr{WMD} and the number of orthogonalization
parameters, using parameter optimization on the \abbr{WMD} with regularized
embeddings is computationally intractable. Therefore, we do not report results
for the \abbr{WMD} with regularized embeddings in
figures~\ref{fig:test-error}--\ref{fig:avg-test-error-ratio}.

\section{Conclusion}
\label{sec:conclusion}
Word embeddings achieve state-of-the-art results on several
\abbr{NLP} tasks, predominantly at the sentence level, but overfitting is a
major issue, especially when applied at the document level with a large number
of words. We have shown that regularization of word embeddings significantly
improves their performance not only on the word analogy, word similarity, and
semantic text similarity tasks, but also on the text classification task. 

We further show that the most effective word embedding regularization technique
is orthogonalization and we prove a connection between orthogonalization, Cholesky
factorization and orthogonalized word embeddings. With word embedding orthogonalization,
the task performance of the soft \abbr{VSM} exceeds the \abbr{WMD}, an earlier
known state-of-the-art document distance measure, while being several orders of
magnitude faster. This is an important step in bringing application of word
embeddings from supercomputers to mobile and embedded systems.

\begin{acknowledgements}
The authors are grateful to their colleagues who helped make the language precise
and engaging, the figures crisp and accessible, and the bibliographical references
valid and consistent.
\end{acknowledgements}

\begin{funding}
\textls[-20]{%
  First author's work was graciously funded by the South Moravian Centre for
  International Mobility as a part of the Brno Ph.D.\ talent project.%
}
\end{funding}

\begin{conflict}
The authors declare that they have no conflict of interest.
\end{conflict}

\bibliography{ayetiran,witiko,sojka}

\begin{thebibliography}{47}
\providecommand{\natexlab}[1]{#1}
\providecommand{\url}[1]{{#1}}
\providecommand{\urlprefix}{URL }
\expandafter\ifx\csname urlstyle\endcsname\relax
  \providecommand{\doi}[1]{DOI~\discretionary{}{}{}#1}\else
  \providecommand{\doi}{DOI~\discretionary{}{}{}\begingroup
  \urlstyle{rm}\Url}\fi
\providecommand{\eprint}[2][]{\url{#2}}

\bibitem[{Agresti and Coull(1998)}]{agresti1998approximate}
Agresti A, Coull BA (1998) Approximate is better than “exact” for interval
  estimation of binomial proportions. The American Statistician 52(2):119--126,
  \doi{10.2307/2685469}

\bibitem[{Bengio et~al.(2003)Bengio, Ducharme et~al.}]{bengio2003neural}
Bengio Y, Ducharme R, et~al. (2003) A neural probabilistic language model.
  Journal of Machine Learning Research 3:1137--1155,
  \urlprefix\url{http://dl.acm.org/citation.cfm?id=944919.944966}

\bibitem[{Benjamini and Hochberg(1995)}]{benjamini1995controlling}
Benjamini Y, Hochberg Y (1995) Controlling the false discovery rate: A
  practical and powerful approach to multiple testing. JRSSB 57(1):289--300,
  \urlprefix\url{http://www.jstor.org/stable/2346101}

\bibitem[{Berend(2018)}]{berend}
Berend G (2018) $\ell_1$ regularization of word embeddings for multi-word
  expression identification. Acta Cybernetica 23(3):801--813,
  \doi{10.14232/actacyb.23.3.2018.5}

\bibitem[{Cardoso-Cachopo(2007)}]{cachopo2007improving}
Cardoso-Cachopo A (2007) Improving methods for single-label text
  categorization. PhD thesis, Instituto Superior T{\'e}cnico, University of
  Lisbon, Portugal,
  \urlprefix\url{http://web.ist.utl.pt/~acardoso/docs/2007-phd-thesis.pdf},
  accessed 23 October 2019

\bibitem[{Charlet and Damnati(2017)}]{charlet2017simbow}
Charlet D, Damnati G (2017) {SimBow} at {SemEval}-2017 task 3: Soft-cosine
  semantic similarity between questions for community question answering. In:
  Proceedings of the 11th International Workshop on Semantic Evaluation
  (SemEval-2017), Association for Computational Linguistics, pp 315--319,
  \doi{10.18653/v1/S17-2051}

\bibitem[{Chen et~al.(2015)Chen, Wilson, Tyree, Weinberger, and
  Chen}]{chen2015compressing}
Chen W, Wilson JT, Tyree S, Weinberger KQ, Chen Y (2015) Compressing neural
  networks with the hashing trick. In: Proceedings of the 32nd International
  Conference on International Conference on Machine Learning, JMLR.org, Lille,
  France, ICML '15, vol~37, pp 2285--2294,
  \urlprefix\url{http://dl.acm.org/citation.cfm?id=3045118.3045361}

\bibitem[{Courbariaux et~al.(2016)Courbariaux, Hubara, Soudry, El-Yaniv, and
  Bengio}]{courbariaux2016binarized}
Courbariaux M, Hubara I, Soudry D, El-Yaniv R, Bengio Y (2016) Binarized neural
  networks: Training deep neural networks with weights and activations
  constrained to $+1$ or $-1$. arXiv preprint
  \urlprefix\url{https://arxiv.org/abs/1602.02830}, accessed 22 October 2019

\bibitem[{Deerwester et~al.(1990)Deerwester, Dumais
  et~al.}]{deerwester1990indexing}
Deerwester S, Dumais ST, et~al. (1990) Indexing by latent semantic analysis.
  Journal of the American Society for Information Science 41(6):391--407,
  \urlprefix\url{https://doi.org/10.1002/(SICI)1097-4571(199009)41:6\%3C391::AID-ASI1\%3E3.0.CO;2-9}

\bibitem[{Garten et~al.(2015)Garten, Sagae et~al.}]{garten2015combining}
Garten J, Sagae K, et~al. (2015) Combining distributed vector representations
  for words. In: Proceedings of the 1st Workshop on Vector Space Modeling for
  Natural Language Processing, Association for Computational Linguistics, pp
  95--101, \doi{10.3115/v1/W15-1513}

\bibitem[{Greene and Cunningham(2006)}]{greene06icml}
Greene D, Cunningham P (2006) Practical solutions to the problem of diagonal
  dominance in kernel document clustering. In: Proceedings 23rd International
  Conference on Machine learning (ICML '06), ACM Press, pp 377--384,
  \doi{10.1145/1143844.1143892}

\bibitem[{Hersh et~al.(1994)Hersh, Buckley, Leone, and Hickam}]{Hersh}
Hersh W, Buckley C, Leone TJ, Hickam D (1994) {OHSUMED}: An interactive
  retrieval evaluation and new large test collection for research. In:
  Proceedings of the 17th Annual International ACM-SIGIR Conference on Research
  and Development in Information Retrieval (SIGIR '94), Springer-Verlag New
  York, Inc., Dublin, Ireland, pp 192--201,
  \urlprefix\url{http://dl.acm.org/citation.cfm?id=188490.188557}

\bibitem[{Hinton(2012)}]{hinton2012neural}
Hinton G (2012) Neural networks for machine learning. Coursera
  \urlprefix\url{https://www.cs.toronto.edu/~hinton/coursera_lectures.html},
  accessed 23 October 2019

\bibitem[{Hinton et~al.(2014)Hinton, Vinyals, and Dean}]{hinton2014distilling}
Hinton G, Vinyals O, Dean J (2014) Distilling the knowledge in a neural
  network. In: NIPS 2014 Deep Learning and Representation Learning Workshop,
  \urlprefix\url{http://www.dlworkshop.org/65.pdf}, accessed 23 October 2019

\bibitem[{Joachims(1998)}]{joachims1998text}
Joachims T (1998) Text categorization with support vector machines: Learning
  with many relevant features. In: Proceedings of European Conference on
  Machine Learning (ECML), Springer, pp 137--142, \doi{10.1007/BFb0026683}

\bibitem[{Joulin et~al.(2016)Joulin, Grave, Bojanowski, Douze, J{\'e}gou, and
  Mikolov}]{joulin2016fasttext1}
Joulin A, Grave E, Bojanowski P, Douze M, J{\'e}gou H, Mikolov T (2016)
  {FastText.zip: Compressing text classification models}. arXiv preprint
  \urlprefix\url{https://arxiv.org/abs/1612.03651}, accessed 22 October 2019

\bibitem[{Kusner et~al.(2015)Kusner, Sun et~al.}]{kusner2015word}
Kusner MJ, Sun Y, et~al. (2015) From word embeddings to document distances. In:
  Proceedings of the 32nd International Conference on International Conference
  on Machine Learning, JMLR.org, vol~37, pp 957--966,
  \urlprefix\url{https://dl.acm.org/citation.cfm?id=3045118.3045221}

\bibitem[{Kuzi et~al.(2016)Kuzi, Shtok, and Kurland}]{kuzi2016query}
Kuzi S, Shtok A, Kurland O (2016) Query expansion using word embeddings. In:
  Proceedings of the 25th ACM International Conference on Information and
  Knowledge Management (CIKM), ACM, pp 1929--1932,
  \doi{10.1145/2983323.2983876}

\bibitem[{Labutov and Lipson(2013)}]{IgorandHod}
Labutov I, Lipson H (2013) Re-embedding words. In: Proceedings of the 51st
  Annual Meeting of the Association for Computational Linguistics, Association
  for Computational Linguistics, Sofia, Bulgaria, pp 489--493,
  \urlprefix\url{https://www.aclweb.org/anthology/P13-2087}

\bibitem[{Lam(2018)}]{lam}
Lam M (2018) {Word2Bits} -- quantized word vectors. arXiv preprint
  \urlprefix\url{https://arxiv.org/abs/1803.05651}, accessed 22 October 2019

\bibitem[{Lang(1995)}]{Lang95}
Lang K (1995) {NewsWeeder}: Learning to filter netnews. In: {Proceedings of the
  12th International Conference on Machine Learning}, Morgan Kaufmann, pp
  331--339, \doi{10.1016/B978-1-55860-377-6.50048-7}

\bibitem[{Lewis(1997)}]{lewis1997reuters}
Lewis DD (1997) Reuters-21578 text categorization test collection. Distribution
  10, AT\&T Labs-Research
  \urlprefix\url{http://www.daviddlewis.com/resources/testcollections/reuters21578/},
  accessed 22 October 2019

\bibitem[{Maaten and Hinton(2008)}]{maaten2008visualizing}
Maaten L, Hinton G (2008) Visualizing data using t-{SNE}. Journal of machine
  learning research 9(Nov):2579--2605

\bibitem[{Mahoney(2011)}]{mahoney11about}
Mahoney M (2011) About the test data.
  \urlprefix\url{http://mattmahoney.net/dc/textdata.html}, accessed 31 October
  2019

\bibitem[{McAuley et~al.(2015)McAuley, Targett et~al.}]{mcauley2015image}
McAuley J, Targett C, et~al. (2015) Image-based recommendations on styles and
  substitutes. In: Proceedings of the 38th International ACM SIGIR Conference
  on Research and Development in Information Retrieval, ACM, pp 43--52,
  \doi{10.1145/2766462.2767755}

\bibitem[{Mikolov et~al.(2013)Mikolov, Chen et~al.}]{mikolov2013efficient}
Mikolov T, Chen K, et~al. (2013) Efficient estimation of word representations
  in vector space. arXiv preprint
  \urlprefix\url{https://arxiv.org/abs/1301.3781v3}, accessed 22 October 2019

\bibitem[{Novotn{\'{y}}(2018)}]{vit}
Novotn{\'{y}} V (2018) Implementation notes for the soft cosine measure. In:
  Proceedings of 27th ACM International Conference on Information and Knowledge
  Management (CIKM '18), Association of Computing Machinery, pp 1639--1642,
  \doi{10.1145/3269206.3269317}

\bibitem[{Pele and Werman(2008)}]{pele2008}
Pele O, Werman M (2008) A linear time histogram metric for improved sift
  matching. In: {Proceedings of the 10th ECML: Part III}, Springer-Verlag,
  Berlin, Heidelberg, ECCV '08, pp 495--508, \doi{10.1007/978-3-540-88690-7_37}

\bibitem[{Pele and Werman(2009)}]{pele2009}
Pele O, Werman M (2009) Fast and robust earth mover's distances. In: IEEE 12th
  International Conference on Computer Vision, IEEE, pp 460--467,
  \doi{10.1109/ICCV.2009.5459199}

\bibitem[{Peng et~al.(2015)Peng, Mou, Li, Chen, Lu, and Jin}]{pengetal}
Peng H, Mou L, Li G, Chen Y, Lu Y, Jin Z (2015) A comparative study on
  regularization strategies for embedding-based neural networks. In:
  Proceedings of the 2015 Conference on Empirical Methods in Natural Language
  Processing, Association for Computational Linguistics, Lisbon, Portugal, pp
  2106--2111, \doi{10.18653/v1/D15-1252}

\bibitem[{Pennington et~al.(2014)Pennington, Socher, and
  Manning}]{pennington2014glove}
Pennington J, Socher R, Manning C (2014) Glove: Global vectors for word
  representation. In: Proceedings of the EMNLP 2014 conference, Association for
  Computational Linguistics, pp 1532--1543, \doi{10.3115/v1/D14-1162}

\bibitem[{Peters et~al.(2018)Peters, Neumann, Iyyer, Gardner, Clark, Lee, and
  Zettlemoyer}]{DBLP:conf/naacl/PetersNIGCLZ18}
Peters ME, Neumann M, Iyyer M, Gardner M, Clark C, Lee K, Zettlemoyer L (2018)
  Deep contextualized word representations. In: Proceedings of the 2018
  Conference of the North American Chapter of the Association for Computational
  Linguistics: Human Language Technologies (NAACL-HLT), Association for
  Computational Linguistics, New Orleans, Louisiana, USA, pp 2227--2237,
  \doi{10.18653/v1/N18-1202}

\bibitem[{Qi et~al.(2018)Qi, Sachan et~al.}]{qi2018and}
Qi Y, Sachan DS, et~al. (2018) When and why are pre-trained word embeddings
  useful for neural machine translation? arXiv preprint
  \urlprefix\url{https://arXiv.org/abs/1804.06323v2}, accessed 22 October 2019

\bibitem[{{\v R}eh{\r u}{\v r}ek and Sojka(2010)}]{rehurek10framework}
{\v R}eh{\r u}{\v r}ek R, Sojka P (2010) Software framework for topic modelling
  with large corpora. In: Proceedings of the LREC 2010 Workshop on New
  Challenges for NLP Frameworks, Valletta, Malta, pp 45--50,
  \doi{10.13140/2.1.2393.1847}

\bibitem[{Robertson(2004)}]{robertson2004understanding}
Robertson S (2004) Understanding inverse document frequency: on theoretical
  arguments for {IDF}. Journal of documentation 60(5):503--520,
  \doi{10.1108/00220410410560582}

\bibitem[{Salton and Buckley(1988)}]{ml:SaltonBuckley1988}
Salton G, Buckley C (1988) Term-weighting approaches in automatic text
  retrieval. Information Processing \& Management 24:513--523,
  \doi{10.1016/0306-4573(88)90021-0}

\bibitem[{Sanders(2011)}]{sanders2011sanders}
Sanders NJ (2011) Sanders-twitter sentiment corpus. Sanders Analytics LLC\ 242,
  \urlprefix\url{http://www.sananalytics.com/lab/twitter-sentiment/}, accessed
  28 March 2018

\bibitem[{See et~al.(2016)See, Luong, and Manning}]{see2016compression}
See A, Luong MT, Manning CD (2016) Compression of neural machine translation
  models via pruning. In: Proceedings of The 20th {SIGNLL} Conference on
  Computational Natural Language Learning, Association for Computational
  Linguistics, Berlin, Germany, pp 291--301, \doi{10.18653/v1/K16-1029}

\bibitem[{Shu and Nakayama(2017)}]{shu2017compressing}
Shu R, Nakayama H (2017) Compressing word embeddings via deep compositional
  code learning. arXiv preprint
  \urlprefix\url{https://arxiv.org/abs/1711.01068}, accessed 22 October 2019

\bibitem[{Sidorov et~al.(2014)}]{sidorov2014soft}
Sidorov G, et~al. (2014) Soft similarity and soft cosine measure: Similarity of
  features in vector space model. CyS 18(3):491--504,
  \doi{10.13053/cys-18-3-2043}

\bibitem[{Singhal et~al.(2001)}]{singhal2001modern}
Singhal A, et~al. (2001) Modern information retrieval: {A} brief overview. IEEE
  Data Engineering Bulletin 24(4):35--43

\bibitem[{Song et~al.(2017)Song, Lee, and Xia}]{songetal}
Song Y, Lee CJ, Xia F (2017) Learning word representations with regularization
  from prior knowledge. In: Proceedings of the 21st Conference on Computational
  Natural Language Learning (CoNLL 2017), Association for Computational
  Linguistics, Vancouver, Canada, pp 143--152, \doi{10.18653/v1/K17-1016}

\bibitem[{Srivastava et~al.(2014)Srivastava, Hinton, Krizhevsky, Sutskever, and
  Salakhutdinov}]{Nitishetal}
Srivastava N, Hinton G, Krizhevsky A, Sutskever I, Salakhutdinov R (2014)
  Dropout: {A} simple way to prevent neural networks from overfitting. Journal
  of Machine Learning Research 15:1929--1958,
  \urlprefix\url{http://dl.acm.org/citation.cfm?id=2627435.2670313}

\bibitem[{Sun et~al.(2016)Sun, Guo, Lan, Xu, and Cheng}]{sunetal}
Sun F, Guo J, Lan Y, Xu J, Cheng X (2016) Sparse word embeddings using $\ell_1$
  regularized online learning. In: Proceedings of the 25th International Joint
  Conference on Artificial Intelligence (IJCAI '16), AAAI Press, pp 2915--2921,
  \urlprefix\url{http://dl.acm.org/citation.cfm?id=3060832.3061029}

\bibitem[{Wu et~al.(2018)Wu, Yen, Xu, Xu, Balakrishnan, Chen, Ravikumar, and
  Witbrock}]{Wuetal2018}
Wu L, Yen IEH, Xu K, Xu F, Balakrishnan A, Chen PY, Ravikumar P, Witbrock MJ
  (2018) {Word Mover’s Embedding: From Word2Vec to Document Embedding}. In:
  Proceedings of the 2018 Conference on Empirical Methods in Natural Language
  Processing, Association for Computational Linguistics, pp 4524--4534,
  \urlprefix\url{https://www.aclweb.org/anthology/D18-1482}

\bibitem[{Yang et~al.(2017)Yang, Lu, and Zheng}]{yangetal}
Yang W, Lu W, Zheng VW (2017) A simple regularization-based algorithm for
  learning cross-domain word embeddings. In: Proceedings of the 2017 Conference
  on Empirical Methods in Natural Language Processing, Association for
  Computational Linguistics, Copenhagen, Denmark, pp 2898--2904,
  \doi{10.18653/v1/D17-1312}

\bibitem[{Zuccon et~al.(2015)Zuccon, Koopman, Bruza, and
  Azzopardi}]{Zucconetal}
Zuccon G, Koopman B, Bruza P, Azzopardi L (2015) Integrating and evaluating
  neural word embeddings in information retrieval. In: Proceedings of the 20th
  Australasian Document Computing Symposium (ADCS), ACM, Parramatta, NSW,
  Australia, pp 12:1--12:8, \doi{10.1145/2838931.2838936}

\end{thebibliography}
\bibliographystyle{spbasic}
\end{document}